\documentclass[a4paper,12pt]{article} 

\usepackage{times}
\usepackage{soul}
\usepackage{url}
\usepackage[hidelinks]{hyperref}
\usepackage[utf8]{inputenc}
\usepackage[small]{caption}
\usepackage{graphicx}
\usepackage{amsmath,amssymb}
\usepackage{amsthm}
\usepackage{booktabs}
\usepackage{algorithm}
\usepackage[noend]{algorithmic}
\usepackage[switch]{lineno}
\usepackage{todonotes}

\usepackage{comment}
\usepackage{xcolor}

\newtheorem{theorem}{Theorem}
\newtheorem{lemma}[theorem]{Lemma}
\newtheorem{corollary}[theorem]{Corollary}

\newtheorem{example}[theorem]{Example}
\newtheorem{definition}[theorem]{Definition}

\newtheorem{claim}[theorem]{Claim}

\DeclareMathOperator\var{Var}

\newcommand{\problemFont}[1]{\textsc{#1}}

\newcommand{\affine}{\protect\ensuremath \mathrm{AFF}}
\newcommand{\ar}{\mathrm{ar}}
\newcommand{\pol}{\protect\ensuremath\text{Pol}}
\newcommand{\inv}{\protect\ensuremath\text{Inv}}\newcommand{\cvred}{\mathbin{\leq^{\mathrm{CV}}}}
\newcommand{\lvred}{\mathbin{\leq^{\mathrm{LV}}}}
\newcommand{\cveq}{\mathbin{=^{\mathrm{CV}}}}
\newcommand{\m}[1]{\mathrm{Mod}(#1)}
\newcommand{\pro}{\mathrm{Proj}}
\newcommand{\cons}{\protect\ensuremath\text{cons}}
\newcommand{\minor}{\mathrm{Min}}

\newcommand{\PMABD}{\protect\ensuremath\problemFont{P-ABD}}
\newcommand{\MABD}{\protect\ensuremath\problemFont{ABD}}
\newcommand{\ABD}{\protect\ensuremath\problemFont{(P-)ABD}}

\newcommand{\SAT}{\protect\ensuremath\problemFont{SAT}}
\newcommand{\simpleSAT}{\protect\ensuremath\problemFont{SimpleSAT}}

\newcommand{\KB}{\protect\ensuremath\text{KB}}

\renewcommand{\P}{\protect\ensuremath\text{P}}
\newcommand{\NP}{\protect\ensuremath\text{NP}}
\newcommand{\coNP}{\protect\ensuremath\text{coNP}}
\newcommand{\SigPtwo}{\protect\ensuremath\Sigma^\text{P}_2}
\renewcommand{\setminus}{\protect\ensuremath-}

\newcommand{\XSAT}{\protect\ensuremath \mathrm{XSAT}}

\newcommand{\Lits}[1]{\protect\ensuremath{{\mathrm{Lits}(#1)}}}

\title{A Fine-Grained Complexity View on Propositional Abduction - Algorithms and Lower Bounds}
\author{
		Victor Lagerkvist,\textsuperscript{\rm 1}
		Mohamed Maizia,\textsuperscript{\rm 1,2}
		Johannes Schmidt\textsuperscript{\rm 2}
\\    
{\footnotesize
    \textsuperscript{\rm 1}
Department of Computer and Information Science, Link\"oping University, Sweden
}
		\\
{\footnotesize
    \textsuperscript{\rm 2} Department of Computer Science and Informatics, J\"onk\"oping University, Sweden
}
		\\
{\footnotesize
		victor.lagerkvist@liu.se, mohamed.maizia@ju.se, johannes.schmidt@ju.se
}
}

\begin{document}

\maketitle

\begin{abstract}
    The {\em Boolean satisfiability problem} (SAT) is a well-known example of monotonic reasoning, of intense practical interest due to fast solvers, complemented by rigorous fine-grained complexity results.
    However, for {\em non-monotonic} reasoning, e.g., {\em abductive} reasoning,  comparably little is known outside classic complexity theory. In this paper we take a first step of bridging the gap between monotonic and non-monotonic reasoning by analyzing the complexity of intractable abduction problems under the seemingly overlooked but natural parameter $n$: the number of variables in the knowledge base. 
    We obtain several positive results for $\Sigma^P_2$- as well as NP- and coNP-complete fragments, which implies the first example of beating exhaustive search for a $\Sigma^P_2$-complete problem (to the best of our knowledge). We complement this with lower bounds and for many fragments rule out improvements under the {\em (strong) exponential-time hypothesis}.
\end{abstract}

\section{Introduction}\label{sec:introduction}


The {\em Boolean satisfiability} problem is a well-known NP-complete problem. Due to the rapid advance of SAT solvers, many combinatorial problems are today  solved by reducing to SAT, which can then be solved with off-the-shelf  solvers. SAT fundamentally encodes a form of {\em monotonic} reasoning in the sense that conclusions remain valid regardless if new information is added. However, the real world is {\em non-monotonic}, meaning that one should be able to retract a statement if new data is added which violates the previous conclusion. One of the best known examples of non-monotonic reasoning is {\em abductive} reasoning where we are interested in finding an explanation, vis-à-vis a hypothesis, of some observed manifestation. Abduction has many practical applications, e.g.,
scientific discovery \cite{InoueSIKN09},
network security \cite{AlbertiCGLMT05},
computational biology \cite{RayAKD06},
medical diagnosis \cite{ObeidOMO19},
knowledge base updates \cite{SakamaI03},
and explainability issues in machine learning and decision support systems \cite{IgnatievNM19,DBLP:conf/ijcai/Ignatiev20,BrardaTG21,RacharakT21}. This  may be especially poignant in  forthcoming decades due to the continued emergence of AI in new and surprising applications, which need to be made GDPR compliant~\cite{SovranoVP20} and explainable. The incitement for solving abduction fast, even when it is classically intractable, thus seems highly practically motivated.

Can non-monotonic reasoning be performed as efficiently as monotonic reasoning, or are there fundamental differences between the two? The classical complexity of abduction (and many other forms of non-monotonic reasoning) is well-known~\cite{eiter1995complexity,DBLP:journals/eatcs/ThomasV10} and suggests a difference: SAT is NP-complete, while most forms of non-monotonic reasoning, including {\em propositional} abduction, are generally $\Sigma^P_2$-complete. However, modern complexity theory typically tells a different story, where classical hardness results do not imply that the problems are hopelessly intractable, but rather that different algorithmic schemes should be applied. For SAT, there is a healthy amount of theoretical research complementing the advances of SAT solvers, and $k$-SAT for every $k$ can be solved substantially faster than $2^n$ (where $n$ is the number of variables) via the resolution-based {\em PPSZ} algorithm~\cite{PPSZ05JACM}. There is a complementary theory of lower bounds where the central conjecture is that 3-SAT is not solvable in $2^{o(n)}$ time ({\em exponential-time hypothesis} (ETH)~\cite{impagliazzo2001}) and the {\em strong exponential-time hypothesis} (SETH) which implies that SAT  with unrestricted clause length (CNF-SAT) cannot be solved in $c^n$ time for any $c < 2$.


In contrast, the precise exponential time complexity of abduction is currently a blind spot, and {\em no} improved algorithms are known for the intractable cases. In this paper we thus issue a systematic attack on the complexity of abduction with a particular focus on the natural complexity parameter $n$, the number of variables in the knowledge base, sometimes supplemented by $|H|$ or $|M|$, the number elements in the hypothesis $H$ or manifestation $M$.
To obtain general results we primarily consider the setup where we are given a set of relations $\Gamma$ (a {\em constraint language}) where the knowledge base of an instance is provided by a $\Gamma$-formula. We write $\MABD(\Gamma)$ for this problem and additionally also consider the variant where an explanation only consists of positive literals ($\PMABD(\Gamma)$) since these two variants exhibit interesting differences.
The classical complexity of abduction is either in P, NP-complete, coNP-complete, or $\Sigma^P_2$-complete \cite{NoZa2008}, and for which intractable $\Gamma$ is it possible to beat exhaustive search? According to Cygan et al., tools to precisely analyze the exponential time complexity of NP-complete problems are in its infancy~\cite{CyganDLMNOPSW16}. For problems at higher levels of the polynomial hierarchy the situation is even more dire. Are algorithmic approaches for problems in NP still usable? Are the tools to obtain lower bounds still usable? 
Why are no sharp upper bounds known for problems in non-monotonic reasoning, and are these problems fundamentally different from e.g. satisfiability problems?
%

%

We successfully answer many of these questions with  novel algorithmic contributions. First, in Section~\ref{section:enum} we show why enumerating all possible subsets of the hypothesis gives a bound $2^n$ for $\MABD$ and the (surprisingly bad) $3^n$ bound for $\PMABD$. Hence, any notion of improvement should be measured against $2^{n}$ for $\MABD$ and $3^n$ for $\PMABD$. Generally improving the factor $2^{|H|}$ (which may equal $2^{n}$) seems difficult but we do manage this for languages $\Gamma$ where all possible models of the knowledge base can be enumerated in $c^n$ time, for some $c \leq 2$, which we call {\em sparsely enumerable} languages. We succeed with this for both $\MABD(\Gamma)$ (Section~\ref{sec:mabd_faster}) and $\PMABD(\Gamma)$ (Section~\ref{sec:pabd_faster}), and while the algorithms for the two different cases share ideas, the details differ in intricate ways. It should be remarked that both algorithms solve the substantially more general problem of enumerating {\em all} (maximal) explanations which may open up further, e.g., probabilistic, applications for abduction. The enumeration algorithms in addition to exponential time also need exponential memory, but we manage to improve the naive $3^n$ bound for $\PMABD(\Gamma)$ to $2^n$ with only polynomial memory. 
The sparsely enumerable property is strong: it fails even for $2$-SAT and it is a priori not clear if it is ever true for intractable languages. Despite this, we manage to (in Section~\ref{sec:sparse}) describe three properties implying sparse enumerability. This captures relations definable by equations $x_1 + \ldots + x_k = q \pmod p$ ($\textsc{Equations}^k$). The problem(s) $\ABD(\textsc{Equations})$ is $\Sigma^P_2$-complete and is, to the best of our knowledge, the first example of beating exhaustive search for a $\Sigma^P_2$-complete problem (under $n$). This yields improved algorithms for $\Sigma^P_2$-complete $\PMABD(\textsc{XSAT})$ ({\em exact satisfiability}) and NP-complete $\PMABD(\affine^{(\leq k)})$ (arity bounded equations over GF(2)).


\begin{table}[h!]
{\small
\centering
\begin{tabular}{@{}lccc@{}}
\toprule
(Type) Class             & Classical complexity & Improved  \\ \midrule
\textsc{Equations}$^k$ ($k \geq 1$)           & $\Sigma^P_2$-C    & Yes     \\
\textsc{XSAT}           & $\Sigma^P_2$-C    & $O^*(2^{\frac{n}{2}})$           \\
(P) \textsc{AFF}$^{k}$ ($k \geq 1$)          & NP-C    & Yes \\
(M) $k$-CNF$^+$ ($k \geq 1$)           & NP-C    & Yes    \\
(P) $k$-CNF$^- \cup \text{IMP}$ ($k \geq 1$)           & NP-C    & Yes \\ 
(P) finite $1$-valid            & coNP-C    & Yes   \\
\bottomrule
\end{tabular}
}
\caption{Upper bounds for $\PMABD$ and $\MABD$.}
\label{tab:upper_bounds}
\end{table}

\begin{table}[h!]
{\small 
\centering
\begin{tabular}{@{}lccc@{}}
\toprule
(Type) Class             & Assumption & Bound  \\ \midrule
(M) 2-CNF$^+$          & ETH    & $(\frac{|H|}{|M|})^{o(|M|)}$ \\
(P) 2-CNF$^+ \cup \text{IMP}$          & ETH    & $(\frac{|H|}{|M|})^{o(|M|)}$ \\
(M) $k$-CNF ($k \geq 4$)           & SETH    & $2^n$    \\
(P) $k$-CNF ($k \geq 4$)           & SETH    & $1.4142^n$    \\
CNF$^- \cup \text{IMP}$, Horn            & SETH   & $1.2599^n$   \\
(M) CNF$^+$, DualHorn            & SETH   & $1.2599^n$   \\
\bottomrule
\end{tabular}
}
\caption{Lower  bounds for $\PMABD$ and $\MABD$.}
\label{tab:lower_bounds}
\end{table}

In Section~\ref{section:alg} we consider more restricted types of abduction problems with a particular focus on $\ABD(k\text{-CNF}^+)$ where $k$-CNF$^+$ contains all positive clauses of arity $k$. Here, the problems are only NP-complete, in which case circumventing the $2^{|H|}$ barrier appears easier. For these, and similar, problems, we construct an improved algorithm based on a novel reduction to a problem 
$\simpleSAT^p$ which can be solved by branching. For coNP, only $\PMABD(\Gamma)$ becomes relevant, and we (in Section~\ref{sec:conp})  prove a simple but general improvement whenever a finite $\Gamma$ is invariant under a constant Boolean operation.

In Section~\ref{sec:lower}, we prove lower bounds under (S)ETH for missing intractable cases. Let $\text{IMP} = \{\{(0,0),$ $(0,1), (1,1)\}\}$. Under the ETH, we first prove that $\MABD(\text{2-CNF}^+)$ and $\MABD(\text{2-CNF}^-\cup\text{IMP})$ cannot be solved in time $(\frac{|H|}{|M|})^{o(|M|)}$ under ETH, which asymptotically matches exhaustive search. For classical cases like $k$-CNF ($k \geq 4$) and NAE-$k$-SAT ($k \geq 5$) we establish sharp lower bounds of the form $2^n$ for $\MABD$ and $1.4142^{n}$ for $\PMABD$ under the SETH. For $\ABD(\text{CNF}^-\cup\text{IMP})$, we rule out improvements to $1.2599^{n}$, $1.4142^{|H|}$, $2^{|M|}$, or $\left(\frac{|H|}{|M|}\right)^{|M|}$ under SETH. This transfers to Horn for $\ABD$ and to $\text{DualHorn}$ for $\MABD$. For $\MABD(2\text{-CNF})$, we prove that sharp lower bounds under the SETH are unlikely unless NP $\subseteq$ P/Poly, leaving its precise fine-grained complexity as an interesting open question.

The results are summarized in Table~\ref{tab:upper_bounds} and Table~\ref{tab:lower_bounds} where (P), respectively (M), indicates that the result only holds for $\PMABD$, respectively $\MABD$. Thus, put together, we obtain a rather precise picture of the fine-grained complexity of $\MABD(\Gamma)$ and $\PMABD(\Gamma)$ for almost all classical intractable languages $\Gamma$. Notably, we have proven that even $\Sigma^P_2$-complete problems can admit improved algorithms with respect to $n$, and that the barrier of exhaustively enumerating all possible explanations can be broken. 

{\small {\em Proofs of statements marked with $(\star)$ can be found in the appendix, following this paper.}}

\section{Preliminaries}\label{prelims}

We begin by introducing basic notation and terminology.

{\bf Propositional logic.}
We assume familiarity with propositional logic, clauses, and formulas in conjunctive/disjunctive form (CNF/DNF). 
We denote $\var(\varphi)$ the variables of a formula $\varphi$ and for a set of formulas $F$, $\var(F) = \bigcup_{\varphi \in F} \var(\varphi)$. We identify finite $F$ with the conjunction of all formulas from $F$, that is $\bigwedge_{\varphi \in F} \varphi$. A \emph{model} of a formula $\varphi$ is an assignment $\sigma: \var(\varphi) \mapsto \{0,1\}$ that satisfies $\varphi$. For two formulas $\psi, \varphi$ we write $\psi \models \varphi$ if every model of $\psi$ also satisfies $\varphi$. For a set of variables $V$,  $\Lits{V}$ the set of all literals formed upon $V$, that is, $\Lits{V} = V \cup \{\neg x \mid x \in V\}$. 

{\bf Boolean constraint languages.}
A \emph{logical relation} of arity $k \in \mathbb{N}$ is a relation $R \subseteq \{0,1\}^k$, and $\ar(R) = k$ denotes the arity.
An ($R$-)constraint $C$ is a formula $C = R(x_1, \dots, x_k)$, where $R$ is a $k$-ary logical relation, and $x_1, \dots, x_k$ are (not necessarily distinct) variables. An assignment $\sigma$ \emph{satisfies} $C$, if $(\sigma(x_1), \dots, \sigma(x_k)) \in R$. A (Boolean) \emph{constraint language} $\Gamma$ is a (possibly infinite) set of logical relations, and a \emph{$\Gamma$-formula} is a conjunction of constraints over elements from $\Gamma$. A $\Gamma$-formula $\phi$ is $\emph{satisfied}$ by an assignment $\sigma$, if $\sigma$ simultaneously satisfies all constraints in it, in which case $\sigma$ is also called a \emph{model of $\phi$}. We define the two constant unary Boolean relations as $\bot = \{(0)\}$ and $\top = \{(1)\}$, and the two constant 0-ary relations as $\sf{f} = \emptyset$ and $\sf{t} = \{\emptyset\}$. 
We say that a $k$-ary relation $R$ is \emph{represented} by a propositional CNF-formula $\phi$ if $\phi$ is a formula over $k$ distinct variables $x_1, \dots, x_k$ and $\phi \equiv R(x_1, \dots, x_k)$. Note such a CNF-representation exists for any logical relation $R \subseteq \{0,1\}^k$.
We write CNF for the relations corresponding to all possible clauses and Horn/DualHorn for the set relations corresponding to Horn/DualHorn clauses. Additionally, $k$-CNF is the subset of CNF of arity $k$, $\text{IMP} = \{R\}$, where $R(x,y) = \{0,1\}^2 -\{(1,0)\}$. 
The notation $\text{(k-)CNF}^+$ (resp. $\text{(k-)CNF}^-$ ) denotes the version where each clause is positive (resp. negative).
We use $\affine$ to denote the set of relations representable by systems of Boolean equations modulo 2, i.e., $x_1 + \ldots + x_k \equiv q \pmod 2$ for $q \in \{0,1\}$. 
As representation of each relation in a constraint language we use the defining CNF-formula unless stated otherwise.

The satisfiability problem for $\Gamma$-formulas, also known as Boolean constraint satisfaction problem, is denoted $\SAT(\Gamma)$. It asks, given a $\Gamma$-formula $\phi$, whether $\phi$ admits a model. 

{\bf Propositional Abduction.}
An instance of the abduction problem over a constraint language $\Gamma$ is given by $(\KB, H, M)$, where $\KB$ is a $\Gamma$-formula, and $H$ and $M$ sets of variables, referred to as {\em hypothesis} and {\em manifestation}.
We let $V = \var(\KB) \cup \var(H) \cup \var(M)$ be the set of variables and write $n = |V|$ for its cardinality. The \emph{symmetric abduction problem}, denoted $\MABD(\Gamma)$, asks whether there exists an $E \subseteq \Lits{H}$ such that 1) $\KB \wedge E$ is satisfiable, and 2) $\KB \wedge E \models M$. If such an $E$ exists, it is called an $\emph{explanation}$ for $M$.
If an explanation $E$ satisfies $\var(E) = H$ it is called a \emph{full explanation}, if it satisfies $E \subseteq H$, it is called a \emph{positive explanation}. If there exists an explanation $E$, it can always be extended to a full explanation (by extending $E$ according to the satisfying assignment underlying condition 1). However, the existence of an explanation does not imply the existence of a positive explanation.
The \emph{positive abduction problem}, denoted $\PMABD(\Gamma)$, asks whether there exists a positive explanation.  We write $\ABD(\Gamma)$ if the specific abduction type is not important.

\begin{example}
    Consider the following example.
    \begin{align*}
        \KB =  \{
        & A \wedge B \rightarrow C, 
         D \rightarrow B, 
         \neg E \rightarrow C \wedge \neg D
        \},\\
        H   =  \{ & A, D, E\} , \hspace{0.5cm}
        M   =  \{C\}
    \end{align*}
$E_1 = \{ A, \neg D, \neg E\}$ is a full explanation, and $E_2 = \{A, D\}$ is a positive explanation.    
\end{example}

 We note at this point that the classical complexity of $\MABD(\Gamma)$ and $\PMABD(\Gamma)$ is fully determined for all $\Gamma$, due to Nordh and Zanuttini \cite{NoZa2008}. An illustration of these classifications is found in the supplementary material.


{\bf Algebra.} We denote by $\Bar{x}$ the Boolean negation operation, that is, $f(x) = \neg x$. An $n$-ary {\em projection} is an operation $f$ of the form $f(x_1, \ldots, x_n) = x_i$ for some fixed $1 \leq i \leq n$.  An operation $f : \{0,1\}^k \rightarrow \{0,1\}$ is \emph{constant} if for all $\textbf{x}\in \{0,1\}^k$ it holds $f(\textbf{x}) = c$, for a $c \in \{0,1\}$. 
For a $k$-ary operation $f \colon \{0,1\}^k \to \{0,1\}$ and $X \subseteq \{0,1\}^k$ we write $f_{\mid X}$ for the function obtained by restricting the domain of $f$ to $X$.
An operation $f$ is a \emph{polymorphism} of a relation $R$ if for every $t_1, \dots, t_k \in R$ it holds that $f(t_1, \dots, t_k) \in R$, where $f$ is applied coordinate-wise. In this case $R$ is called \emph{closed} or \emph{invariant}, under $f$, ard $\inv(F)$ denotes the set of all relations invariant under every function in $F$. Dually, for a set of relations $\Gamma$, $\pol(\Gamma)$ denotes the set of all polymorphisms of $\Gamma$. Sets of the form $\pol(\Gamma)$ are known as \emph{clones}, and sets of the form $\inv(F)$ are known as \emph{co-clones}. A \emph{clone} is a set of functions closed under 1) functional composition, and 2) projections (selecting an arbitrary but fixed coordinate). For a set $B$ of (Boolean) functions, $[B]$ denotes the corresponding clone, and $B$ is called its \emph{base}. There is an inverse relationship between $\pol(\Gamma)$ and $\inv(F)$ but we defer the details to the supplemental material.

{\bf Complexity theory.} We assume familiarity with basic notions in classical complexity theory (cf.~\cite{Sipser97}) and use complexity classes $\P$, $\NP$, $\coNP$, $\NP^\NP = \SigPtwo$.
In this paper we work in the setting of {\em parameterized} complexity where the \emph{complexity parameter} is the number of variables $n$, in either an abduction or SAT instance, or the number of vertices in a graph problem. For a \emph{variable problem} $A$ we let $I(A)$ be the set of instances and $\var(I)$  the set of variables. If clear from the context we usually write $n$ rather than $|\var(I)|$.
We define the following two types of reductions~\cite{jonsson2017} (note that ordinary polynomial-time reductions do not necessarily preserve the number of variables).

\begin{definition}
Be $A,B$ two variable problems. A function $f \colon I(A) \rightarrow I(B)$ is a \emph{many-one linear variable reduction} (LV-reduction) with parameter $C \geq 0$ if $I$ is a yes-instance of $A$ iff $f(I)$ is a yes-instance of $B$,  $|\var(f(I))| = C \cdot |\var(I)|+O(1)$, and $f$ can be computed in polynomial time.
\end{definition}

If in an LV-reduction the parameter $C$ is $1$, we speak of a \emph{CV-reduction}, and we take us the liberty to view a reduction which actually shrinks the number of variables ($|\var(f(I))| \leq |\var(I)| +O(1)$) as a CV-reduction, too. We use $A \cvred B$, respectively $A \lvred B$ as shorthands, and sometimes write $A \cveq B$ if $A \cvred B$ and $B \cvred A$.
For  algorithms' (exponential) running time and space usage we adopt the $O^*$ notation which suppresses polynomial factors. 
A CV-reduction transfers exact exponential running time: if $A \cvred B$, and $B$ can be solved in time $O^*(c^n)$, then also $A$ can be solved in time $O^*(c^n)$ (where $n$ denotes the complexity parameter). LV-reductions, on the other hand, preserve {\em subexponential complexity}, i.e., if $B$ can be solved in $O^*(c^n)$ time for every $c > 1$ and $A \lvred B$ then $A$ is solvable in $O^*(c^n)$ time for every $c > 1$, too. For additional details we refer the reader to~\cite{tcs2021}.

\section{Upper Bounds via SAT Based Approaches}
\label{section:enum}

We begin by investigating the possibility of solving $\MABD(\Gamma)$ and $\PMABD(\Gamma)$ faster, conditioned by a reasonably efficient algorithm for $\SAT(\Gamma)$. This assumption is necessary to beat exhaustive search since  $\SAT(\Gamma) \cvred \ABD(\Gamma)$, which implies that if $\SAT(\Gamma)$ is not solvable in $O(c^n)$ time for any $c < 2$ then $\ABD(\Gamma)$ is not solvable in $O(c^n)$ time for any $c < 2$, either. 


For a constraint language $\Gamma$ we let $\Gamma^+ = \Gamma \cup \{\bot, \top\}$, i.e., $\Gamma$ expanded with the two constant Boolean relations. Trivially, we have $\SAT(\Gamma) \cvred \SAT(\Gamma^+)$, and we observe that if $\pol(\Gamma)$ does not contain a constant operation then we additionally have $\SAT(\Gamma^+) \cvred \SAT(\Gamma)$~\cite{lagerkvist2020f}.  We obtain the following baseline bound to beat.

\begin{theorem}($\star$) \label{thm:brute_force}
Let $\Gamma$ be a constraint language such that $\SAT(\Gamma^+)$ is solvable in $f(n)$ time for some computable function $f \colon \mathbb{N} \to \mathbb{N}$. Then 
\begin{enumerate}
    \item $\MABD(\Gamma)$ is solvable in  $2^{|H|} \cdot f(n - |H|) \cdot (|M| + 1)$ time, 
    \item $\PMABD(\Gamma)$ is solvable in  $\Sigma_{E \subseteq H}f(n - |E|) \cdot (|M| + 1)$ time.
\end{enumerate}
\end{theorem}    

Since every $\SAT(\Gamma^+)$  problem is solvable in $O^*(c^n)$ time for some $c \leq 2$ we
get an overall bound of $2^{|H|} \cdot O^*(c^{n - |H|}) \in O^*(2^n)$ for $\MABD(\Gamma)$, and  $\Sigma_{E \subseteq H} O^*(c^{n - |E|}) = O^*(c^{n - |H|} \cdot (c+1)^{|H|}) \in O^*((c+1)^n) \in O^*(3^n)$ for $\PMABD(\Gamma)$.


The question is now when these baseline bounds can be beaten. As a general method we (for both $\MABD$ and $\PMABD$) consider the assumption that all models of the knowledge base can be enumerated sufficiently fast. 

\begin{definition}
    The set of models of a $\SAT(\Gamma)$ instance $\varphi$  is denoted $\m{\varphi}$. If there exists $c < 2$ such that $|\m{\varphi}| \leq c^n$ then $\SAT(\Gamma)$ is said to be {\em sparse}.
\end{definition}


We also need the corresponding computational property where we require all models to be enumerable fast.

\begin{definition}
    Let $\Gamma$ be a constraint language. If $\m{\varphi}$, for every $\SAT(\Gamma)$ instance $\varphi$, can be enumerated in $O^*(c^n)$ time for some $c < 2$ then we say that $\SAT(\Gamma)$ is {\em sparsely enumerable}.
\end{definition}


\subsection{Faster Algorithms for $\MABD$}
\label{sec:mabd_faster}
We handle $\MABD(\Gamma)$ first since the analysis is simpler than for $\PMABD(\Gamma)$ (in Section~\ref{sec:pabd_faster}).
Trading polynomial for exponential space we consider a faster algorithm for $\MABD(\Gamma)$ under the condition that $\SAT(\Gamma)$ is sparsely enumerable. 

We first state a technical lemma, facilitating the presentation of the algorithms and reductions throughout the following sections.

\begin{lemma}($\star$)\label{lem:preprocessing}
    When solving an abduction instance $(\KB, H, M)$, we can W.L.O.G. assume that $M,H \subseteq \var(\KB)$, via a polynomial time preprocessing. This constitutes even a CV-reduction of the problem to itself.
\end{lemma}

Now the basic idea to solve $\MABD(\Gamma)$ is to define an equivalence relation $\equiv_H$ over $\m{\KB}$ by letting $f \equiv_H g$ if and only if $f_{\mid H} = g_{\mid H}$ (where $H$ is the hypothesis set). We then construct the equivalence classes of $\equiv_H$ and discard a class when it fails to explain $M$.
An explanation exists if there is a non-empty class where every member satisfies $M$. Initially, this requires exponential space, $O^*(2^n)$. However, by only storing information on whether a potential explanation $E$ has an extension that fails to satisfy $M$, space usage is reduced to $O^*(2^{|H|})$. The space usage is further limited by the enumerating algorithm's runtime, $O^*(c^n)$, resulting in total space usage bounded by $O^*(\min(c^n,2^{|H|}))$. We obtain the following theorem.

\begin{theorem}($\star$)
    Let $\Gamma$ be a constraint language where $\SAT(\Gamma)$ is sparsely enumerable in $O^*(c^n)$ time. Then $\MABD(\Gamma)$ is solvable in $O^*(c^n)$ time and $O^*(\min(c^n,2^{|H|}))$ space. 
\end{theorem}

\subsection{Faster Algorithms for $\PMABD$} \label{sec:pabd_faster}
In this section we present improved algorithms for brute force and enumeration for positive abduction, that have the same complexities as the symmetric variants described above. Recall that the baseline bound to beat for $\PMABD$ (from Theorem~\ref{thm:brute_force}) is $O^*(3^n)$, and we begin by lowering this to $O^*(2^n)$ via a more sophisticated exhaustive search scheme.

\begin{algorithm} 
    \caption{Algorithm $\mathcal{A}$ for $\PMABD(\Gamma)$.}
    \label{alg:PABD}
    \begin{algorithmic}[1]
        \REQUIRE $\KB, H,M,D,\delta$
                \STATE $E = D \cup \delta$
				\STATE $G = E \cup \{\neg x \mid x\in H-E\}$
        \IF{$\KB \wedge G \wedge \neg M$ is satisfiable}
           \RETURN $\bot$
        \ENDIF
        \IF{$\KB \wedge G$ is satisfiable}
            \RETURN $\top$
        \ELSE
           \IF{$D = \emptyset$}
             \RETURN $\bot$
           \ENDIF
           \STATE flag = $\bot$
           \FOR{$x \in D$}
             \STATE flag = flag $\vee$ $\mathcal{A}(\KB,H,M,D - \{x\},\delta)$
             \STATE $\delta = \delta \cup \{x\}$
           \ENDFOR
            \RETURN flag
        \ENDIF
				
    \end{algorithmic}
\end{algorithm}

\begin{theorem}$(\star)$
For any constraint language $\Gamma$, $\PMABD(\Gamma)$ can be solved in $O^*(2^n)$ time and polynomial space.
\end{theorem}

\begin{proof}
    Consider Algorithm~\ref{alg:PABD}. The starting parameters are $D = H$ and $\delta = \emptyset$.
The recursive algorithm systematically explores all subsets of $H$ as candidates.
It starts with the base candidate $E = H$. Inside each recursive call, it first checks if the current candidate $E$, extended to a full candidate $G$, entails the manifestation (line 3). If this fails, it concludes that neither $G$ nor any subset of $G$ (which includes $E$ and all its subsets) can be an explanation, and returns $\bot$. Else, it checks if $G$ is consistent with $\KB$. If yes, then $G$ is obviously a (full) explanation, but the algorithm concludes that in this case even $E \subseteq G$ is a (positive) explanation. If $G$ is not consistent with $\KB$, the algorithm concludes that neither $E$ is consistent with $\KB$, and can thus not be an explanation. 
Then the algorithm systematically checks candidates where exactly one variable is removed from $E$ (lines 11--14). Descending in the recursive calls (line 12) it makes sure to systematically explore \emph{all} subsets of a candidate, thereby avoiding visiting the same subset multiple times (this is the purpose of $\delta$). For correctness, we refer to the supplementary material.
\end{proof}

We now consider an algorithm based on enumerating all models of the knowledge base, similar to the symmetric abduction case, by trading polynomial for exponential space.

\begin{theorem}$(\star)$
    Let $\Gamma$ be a constraint language where $\SAT(\Gamma)$ is sparsely enumerable in $O^*(c^n)$ time. Then $\PMABD(\Gamma)$ is solvable in $O^*(c^n)$ time and $O^*(c^n)$ space.
\end{theorem}

\subsection{Provably Sparse Languages}
\label{sec:sparse}

We have proved that $\ABD$ can be solved faster if all models of the underlying SAT problem can be enumerated sufficiently fast. Hence, it is highly desirable to classify the SAT problems where this is indeed the case --- provided that any positive, non-trivial examples actually exist. We obtain a general characterization of such languages based on three abstract properties. Here, we always assume that a $\SAT(\Gamma)$ instance is represented by listing all tuples in a relation.

\begin{definition}
A constraint language $\Gamma$ is {\em asymptotically sparse} if there exists $r_0 \geq 1$ and $c < 2$ such that for $r$-ary $R \in \Gamma$, $r_0 \leq r$ we have $|R| \leq c^{r}$.
\end{definition}

If $R$ is an $n$-ary relation and $g \colon [m] \to [n]$, for some $m \leq n$, then the relation $R_g(x_1, \ldots, x_m) \equiv R(x_{g(1)}, \ldots, x_{g(m)})$ is said to be a {\em minor} of $R$. 

\begin{definition}
    For a constraint language $\Gamma$ we let $\minor(\Gamma) = \{M \mid M \text{ is a minor of } R \in \Gamma\}$ be the set of minors of $\Gamma$.
\end{definition}

Similarly, if $R \in \Gamma$ of arity $\ar(R) = k$ and $f \colon X \to \{0,1\}$ for some $X \subseteq [k]$ then we define the {\em substitution} of $f$ over $R$ as the relation $R_{\mid f} = \{\pro_{[k] - X}(t) \mid t = (c_1, \ldots, c_k) \in R, i \in X \Rightarrow c_i = f(i)\}$ where $\pro_{[k] - X}(t)$ denotes the $(k - |X|)$-ary tuple obtained by only keeping indices in $[k]$ outside $X$. 

\begin{definition}
    A language $\Gamma$ is said to be {\em closed under branching} if it
is closed under substitutions, i.e., if $R \in \Gamma$ and $f \colon X \to \{0,1\}$ for $X \subseteq [\ar(R)]$ then $R_{\mid f} \in \Gamma$, and
    is closed under minors, i.e., $\minor(\Gamma) = \Gamma$.
\end{definition}

We observe that $\Gamma$ is finite if and only if $\minor(\Gamma)$ is finite.

\begin{example} \label{ex:branching}
    Any language $\Gamma$ can be closed under branching simply by repeatedly closing it under minors and substitutions. For example, consider 3-SAT and a positive clause corresponding to the relation $R = \{0,1\}^3 - \{(0,0,0)\}$. Then $\minor(\{R\}) = \{R, \{0,1\}^2 - \{(0,0)\}, \top\}$. If we close $R$ under substitutions then we obtain $\{\{0,1\}^2 - \{(0,0)\}, \top, \sf{f}\}$ by identifying one or more variable to $0$ and $\{\{0,1\}^2, \{0,1\}, \sf{t}\}$ by identifying one or more variable to $1$. 
\end{example}


Thus, while it is easy to close a language $\Gamma$ under branching, this process might introduce undesirable relations of the form $\{0,1\}^k$ which do not enforce any constraints. 

\begin{definition}
    A relation $R \subset \{0,1\}^k$ is said to be {\em non-trivial}. The relation $\{0,1\}^k$ is said to be {\em trivial}.
\end{definition}

By combining these properties we obtain a novel characterization of sparsely enumerable languages. 

\begin{theorem}$(\star)$
    Let $\Gamma$ be a constraint language which (1) is asymptotically sparse, (2) is closed under branching, and (3) every $R \in \Gamma$ is non-trivial. Then $\SAT(\Gamma)$ is sparsely enumerable.
\end{theorem}

We continue by proving that such languages actually exist. First, say that a $k$-ary Boolean relation $R$ is {\em totally symmetric}, or just
{\em symmetric}, if there exists a set $S \subseteq [k] \cup \{0\}$ 
such that $(x_1, \ldots, x_k) \in R$ if and only if $x_1 + \ldots
+ x_k \in S$. Given $S \subseteq \{0, \ldots, k\}$ we write $R_S$ for the symmetric relation induced by $S$, i.e., $R_S(x_1,\ldots,x_k) \equiv (\sum x_i \in S)$.

\begin{definition}
We let $\textsc{Equations} = \{R_S \mid k \geq 1, p,q \leq k + 1, S=\{i \in [k] \cup \{0\} \mid i \equiv q \pmod p\}$.
\end{definition}

Thus, each relation in $\textsc{Equations}$ can be defined by an equation of the form $x_1 + \ldots + x_k = q \pmod p$ for fixed $p,q,k$. In particular $\affine \subseteq \textsc{Equations}$ but it also contains relations inducing $\Sigma^P_2$-complete $\ABD$ problems.
Perhaps contrary to intuition, $\textsc{Equations}$ is {\em not} closed under minors, since the resulting relations are not necessarily symmetric, but
a simple work-around is to fix a finite subset of $\textsc{Equations}$ and then close it under minors. We obtain the following.

\begin{lemma}$(\star)$
The following statements are true.
\begin{enumerate}
    \item 
        $\textsc{Equations}$ is closed under substitutions and contains only non-trivial relations.
\item             Let $k \geq 1$. For $\mathcal{E}^{\leq k} = \{R \in \textsc{Equations} \mid \ar(R) \leq k\}$ $\SAT(\minor(\mathcal{E}^{\leq k}))$ is sparsely enumerable.
\end{enumerate}
\end{lemma}

Finally, let $\XSAT \subseteq \textsc{Equations}$ be the set of relations $R_{1/k}$ representable by equations $x_1 + \ldots + x_k = 1 \pmod k + 1$, and for each $k \geq 1$ the constant relation $\bot^k = \{(0, \ldots, 0)\}$ of arity $k$ (note also that $R_{1/1} = \top$), and, finally, the two nullary relations $\sf{f}$ and $\sf{t}$. The resulting problem $\SAT(\XSAT)$ is thus the natural generalization of the well-known NP-complete problem \textsc{1-in-3-SAT} to arbitrary arities.
Also, recall that $\affine \subseteq \textsc{Equations}$ is the set of relations representable by systems of Boolean equations modulo 2, i.e., $x_1 + \ldots + x_k \equiv q \pmod 2$ for $q \in \{0,1\}$. We additionally write $\affine^{\leq k}$ for the set of affine relations of arity at most $k$.
\begin{theorem}$(\star)$ 
The following statements are true.
\begin{enumerate}
\item 
$\SAT(\XSAT)$ is sparsely enumerable in $O^*(\sqrt{2}^n)$ time.
\item 
$\SAT(\affine^{\leq k})$ is sparsely enumerable for every $k \geq 1$.
\end{enumerate}
\end{theorem}

Let us also observe that the bound for $\SAT(\XSAT)$ is tight in the sense that 
$|\m{\varphi}| = \sqrt{2}^n$, $n = 2m$, if $\varphi$ encodes inequalities between $x_1$ and $x_2$, $x_3$ and $x_4$, and so on.

\subsection{Algorithms for NP-complete fragments}
\label{section:alg}

Symmetric abduction is NP-complete for k-CNF$^+$ languages for any $k \geq 2$, and both symmetric and positive abduction are NP-complete for languages of the form 
k-CNF$^- \cup $ IMP, $k \geq 2$  \cite{NoZa2008}. We show in the following that these cases can be solved in improved time.

\begin{definition}
Denote by $\simpleSAT^p$ the SAT-problem where we are given a formula $\varphi$ of the following form. Here, $C$ stands for a positive clause of size at most $p$ and $T$ stands for a negative term of any size.

$$\varphi =  \bigwedge C \wedge \bigwedge \bigvee T$$

\end{definition}

\begin{lemma}
$\simpleSAT^p$ can be solved in time $O^*(c^n)$, where $c$ is the branching factor associated with a $(1, \dots, p)$-branching.
\end{lemma}
\begin{proof}
Perform a branch and reduce scheme. Variables not occurring in any positive clause can be reduced to 0 (thereby simplifying some of the negative terms). Then branch on the variables of positive clauses, with the standard $(1, \dots, p)$-branching.
\end{proof}

We are now ready to show that $\MABD(\text{k-CNF}^+)$ can be CV-reduced to $ \simpleSAT^k$.

\begin{theorem}($\star$)
$\MABD(\text{k-CNF}^+) \cvred \simpleSAT^k$. Moreover, the $\simpleSAT$-instance contains only variables from $H$.
\end{theorem}

We show that $\ABD(\text{k-CNF}^- \cup \text{IMP})$ can be CV-reduced to $\MABD(\text{k-CNF}^+)$.

\begin{lemma}($\star$)\label{lem:cv_reduction}
$\ABD(\text{k-CNF}^- \cup \text{IMP}) \cvred \MABD(\text{k-CNF}^+)$.
\end{lemma}

The following corollary finally states the improved results, following immediately from the previous statements.
\begin{corollary}
$\MABD(\text{k-CNF}^+)$ and $\ABD(\text{k-CNF}^- \cup \text{IMP})$ can be solved in improved time.
 That is, in time $O^*(c^{|H|})$, for a $c < 2$, stemming from the branching vector $(1, \dots, k)$.
\end{corollary}

\subsection{Algorithms for coNP-complete fragments}
\label{sec:conp}

In the case of positive abduction 
$\coNP$-complete cases arise \cite{NoZa2008} when $\Gamma$ is 1-valid. Under certain additional assumptions, $\PMABD(\Gamma)$ can then be solved in improved time. 

\begin{theorem}\label{thm:coNPcCase}
    Let $\Gamma$ be a 1-valid constraint language. If $\SAT(\Gamma^+)$ can be decided in $O^*(c^n)$ time for $c \leq 2$ then $\PMABD(\Gamma)$ can be decided in $O^*(c^n)$ time.
\end{theorem}
\begin{proof}
First, note that the 1-valid property is responsible for coNP-membership: there is an explanation iff $H$ is an explanation.
Then $\KB \wedge H$ is always consistent (since 1-valid). Thus, we only need to check whether $\KB \wedge H \models M$.
This implication can be decided by invoking the given $\SAT(\Gamma^+)$ algorithm $|M|$ times: $\KB \wedge H \models M$ iff for each $m \in M$ the $\Gamma^+$-formula $\KB \wedge H \wedge \neg m$ is unsatisfiable, and we thus only increase the $O^*(c^n)$ complexity by a polynomial factor.
\end{proof}

\noindent
An  application example is a knowledge base in k-CNF (k $\geq 3$) where each clause contains at least 1 positive literal (this ensures 1-validity). The underlying constraint language $\Gamma$ is still expressive enough to render $\PMABD(\Gamma)$ coNP-hard \cite{NoZa2008}. Then, any $\Gamma^+$-formula is expressible as k-CNF formula (without additional variables), so $\SAT(\Gamma^+)$ can be solved in improved time via the standard $(1, \dots, k)$-branching. With Theorem~\ref{thm:coNPcCase} we conclude that $\PMABD(\Gamma)$ can be solved in improved time.



\section{Lower bounds and Reductions}
\label{sec:lower}

We continue by matching our new upper bounds with lower bounds. We base our lower bounds on ETH and its stronger variant SETH (recall the definition in Section~\ref{sec:introduction}).

\subsection{ETH Based Lower Bounds}
\label{sec:eth}

Our aim in this section is to prove the following theorem.

\begin{theorem} \label{thm:2_cnfplus}
    $\MABD(\text{2-CNF}^+)$ and $\MABD(\text{2-CNF}^-\cup\text{IMP})$ cannot be solved in time $(\frac{|H|}{|M|})^{o(|M|)}$ under ETH.
\end{theorem}

\begin{proof}
We provide a CV-reduction from the {\em $k$-colored clique} problem to $\MABD(\text{2-CNF}^-\cup\text{IMP})$, i.e., given a graph $G=(V,E)$ where the vertices are colored with $k$ different colors,  decide if a clique containing a vertex of each color exists. It is known that $k$-colored clique cannot be solved in time $(\frac{n}{k})^{o(k)}$ under ETH, where $n$ is the number of vertices, and $k$ is the number of colors~\cite{DBLP:journals/eatcs/LokshtanovMS11}.

For the reduction, assume an arbitrary instance of a $k$-colored clique problem over a graph $G=(V,E)$, and where $C={c_i \mid i \in [1 \ldots k]}$ is the set of colors. For each color $c_i$ we add a manifestation $m_i$ to $M$. For each vertex $v_i$ colored with the color $c_i$ we add to $\KB$ the clause $(\neg v_i \lor m_i)$. For each two vertices $v_i$ and $v_j$ that are not connected by an edge, we add the clause $(\neg v_i \lor \neg v_j)$. This completes the reduction. Since the number of variables in the abduction problem is equal to the number of vertices plus the number of colors, this is a CV-reduction. If we can choose a clique that contains at least one vertex from each color without choosing two vertices that are not connected, then for each chosen vertex $v_i$, we set $v_i = \top$ in the abduction instance as part of the solution. This will entail all $m_i$'s without causing a contradiction in $\KB$. The correctness proof from the other side is exactly the same.

Last, by Lemma~\ref{lem:cv_reduction} we additionally obtain that $\MABD(\text{2-CNF}^-\cup\text{IMP})\cvred\MABD(\text{2-CNF}^+)$.
\end{proof}

Note that this lower bound is given with respect to $|H|$ and $|M|$ and not $n$, and it is identical to the running time of the brute force algorithms of these problems. Indeed, in the worst case, we try all possible combinations of hypotheses for all manifestations. If a hypothesis appears for two manifestations at the same time, it only reduces the need to choose an additional hypothesis for the second one. The worst case is then when all hypotheses are split up among manifestations. The number of combinations is thus $(\frac{|H|}{|M|})^{|M|}$, making the algorithm close to optimal.

\subsection{Lower Bounds for $\ABD(\text{4-CNF})$ under SETH}
\label{sec:4_cnf}

Here, we prove a strong lower bound under SETH which shows that we should only expect small improvements for 4-CNF, and, more generally, $k$-CNF for any $k \geq 4$.

\begin{theorem}$(\star)$ \label{thm:4cnf_seth}
Under SETH, there is no $c<1$ such that $\MABD(\text{4-CNF})$ is solvable in $2^{cn}$ time.
\end{theorem}

The following lemma gives the slightly worse lower bound for positive abduction variant of 4-CNF.

\begin{lemma}$(\star)$ 
 Under SETH, there is no $c<1$ such that $\PMABD(\text{4-CNF})$ can be solved in time $1.4142^{cn}$.
\end{lemma}

We remark that this leaves 3-CNF as an interesting open case. We have no algorithms that run in less than $2^n$ so far, and no known lower bounds under SETH.

\subsection{Languages closed under complement}
\label{sec:comp}

We analyze languages closed only under complement, i.e., languages $\Gamma$ such that $\pol(\Gamma) = [\Bar{x}]$. Well-known examples of such languages include {\em not-all-equal} satisfiability. For this class of languages we manage to relate them to languages closed only under projections, in a very precise sense, and show that one without loss of generality can concentrate solely on the latter class.

\begin{theorem}($\star$) \label{thm:in2tobr}
Let $\Gamma$ be a constraint language such that
    $\pol(\Gamma) = [\Bar{x}]$. Then $\ABD(\Gamma \cup \{\bot, \top\}) \cveq \ABD(\Gamma)$.
\end{theorem}

For the next theorem, let $k$-NAE be the set of all $k$-ary relations that forbid exactly two complementary assignments (determined by a sign pattern). The full details can be found in the supplementary material.


\begin{theorem}($\star$)
    $\ABD(\text{k-CNF}) \cvred \ABD(\text{(k+1)-NAE})$.
\end{theorem}

Thus, in general we should not expect to solve complementative abduction problems faster than $2^n$ (via Theorem~\ref{thm:4cnf_seth}).

\subsection{Lower Bounds for Horn and CNF$^+$ Abduction}
\label{sec:horn}

Last, we prove strong lower bounds for Horn and CNF$^+$. These lower bounds do not conclusively rule out algorithms faster than $2^n$ but we stress that sharp lower (under SETH) and upper bounds of the form $c^n$ are rather uncommon in the literature, and an improved upper bound would likely need to use completely new ideas.

\begin{lemma}($\star$)\label{lem:lowerbound-cnf-minus-imp}
    Under SETH, there is no $c<1$ such that (P-)$\MABD(\text{CNF}^-\cup\text{IMP})$ can be solved in time $1.2599^{cn}$ or $1.4142^{c|H|}$ or $2^{c|M|}$ or $\left(\frac{|H|}{|M|}\right)^{c|M|}$.
\end{lemma}

Since $\text{2-CNF}^-\cup\text{IMP}$ is a special case of Horn, we obtain the same lower bounds for Horn knowledge bases. Furthermore, using the reduction from Lemma~\ref{lem:cv_reduction}, we obtain a CV-reduction from (P-)$\MABD(\text{CNF}^-\cup\text{IMP})$ to
$\MABD(\text{CNF}^+)$, and $\text{CNF}^+$ in turn is a special case of $\text{DualHorn}$. We therefore obtain the following Corollary.

\begin{corollary}\label{lem:HornLowerBound}
    Under SETH, there is no $c<1$ such that (P-)$\MABD$(Horn) as well as $\MABD(\text{CNF}^+)$ and $\MABD(\text{DualHorn})$ can be solved in time $1.2599^{cn}$ or $1.4142^{c|H|}$ or $2^{c|M|}$ or $\left(\frac{|H|}{|M|}\right)^{c|M|}$.
\end{corollary}

An interesting question that might occur is whether or not we can have SETH lower bounds for the problems in Section 4.1. Those problems instead have weaker lower bounds from ETH. The following lemma will show that we are unlikely to obtain such bounds with the same method we obtained them for the problems in this Section (4.4).

\begin{lemma}$(\star)$
    If CNF-SAT $\lvred \MABD(\text{2-CNF})$ then NP $\subseteq$ P/Poly.
\end{lemma}

\begin{proof}
    We provide a short proof sketch.
    From~\cite{lance}, we know that CNF-SAT does not admit a kernel of size $f(n)$ where $f$ is a polynomial, i.e., an equisatisfiable instance with e.g.\ a quadratic number of clauses (unless $\NP \subseteq$ P/Poly). To prove the claim it suffices to have a reduction from $\MABD(\text{2-CNF})$ to CNF-SAT that only has a polynomial number of clauses in terms of $n$. Assume an arbitrary $\MABD(\text{2-CNF})$ instance. For every $m_i \in M$, for all clauses of the type: $(\ell_j\lor m_i), \ldots, (\ell_k \lor m_i)$ where $\ell_i \ldots \ell_k$ are a literals, we delete these clauses and add the following CNF clause $(\ell_i \lor \ldots \lor \ell_k)$, and otherwise keep everything unchanged.
\end{proof}


It is thus unlikely to obtain SETH lower bounds for $\MABD(\text{2-CNF})$ from CNF-SAT under LV- or CV-reductions.

\section{Concluding Remarks}
We demonstrated that non-monotonic reasoning, and in particular propositional abduction, for many non-trivial cases {\em do} admit improvements over exhaustive search. We find it particularly interesting that even $\Sigma^P_2$-complete problems fall under the scope of our methods. Might it even be the case that $\Sigma^P_2$ is not such an imposing barrier as classical complexity theory tells us? Nevertheless, despite many positive and negative results, there are still open cases remaining and many interesting directions for future research.

\paragraph*{Faster Enumeration?} We proved that finite subsets of $\affine$ and $\textsc{Equations}$ are susceptible to enumeration. It is easy to see that $\pol(\affine)$ contains the {\em Maltsev} operation $x - y + z (\pmod 2)$, while $\textsc{Equations}$ is exactly the set of symmetric relations invariant under a {\em partial Maltsev} operation~\cite{lagerkvist2022a}. Is it a coincidence that all of our positive enumeration cases can be explained by partial Maltsev operation, or could universal algebra be applied even further? For example, it is straightforward to show that if a language is {\em not} preserved by partial Maltsev, then it can not be sparsely enumerable. Extending this further, if one allows e.g. a polynomial-time preprocessing, could it even be the case that a Boolean (possibly non-symmetric) language is sparsely enumerable if and only if it is invariant under partial Maltsev?

\paragraph*{2- and 3-CNF.} While $\ABD(4\text{-CNF})$ is unlikely to admit improved algorithms, $\ABD(k\text{-CNF})$ for $k \leq 3$ is wide open. These languages are not sparsely enumerable so they do not fall under the scope of the enumeration algorithms, yet, it appears highly challenging to prove sharp lower bounds for them (and recall that CNF-SAT does not admit an LV-reduction to $\ABD(2\text{-CNF})$ unless NP $\subseteq$ P/Poly). As a possible starting point one could consider instances with only a linear number $m$ of clauses  in the knowledge base, or, the more restricted case when each variable may only occur in a fixed number of constraints. Could instances of this kind be solved with enumeration? 

\paragraph*{Other Parameters?} Related to the above question one could more generally ask when $\ABD(\Gamma)$ admits an improved algorithm with complexity parameter $m$, which we observe in general can be much larger than $n$. Do any of the algorithmic results carry over, and can lower bounds be obtained? For the related quantified Boolean formula problems, Williams~\cite{10.5555/545381.545421} constructed an $O(1.709^m)$ time algorithm, so one could be cautiously optimistic about analyzing abduction with $m$.


\newpage

\newpage

\setcounter{section}{0}

\section*{Appendix}

\noindent
For completeness' sake we include the full preliminaries, although most of these notions have already been defined in the main paper. Section~2 contains all missing proofs.

\section{Preliminaries}

\paragraph{Propositional logic}
We assume familiarity with propositional logic. A \emph{literal} is a variable $x$ or its negation $\neg x$. A \emph{clause} is a disjunction of literals and a \emph{term} is a conjunction of literals. A formula $\varphi$ is in \emph{conjunctive normal form} (CNF) if it is a conjunction of clauses, and in \emph{disjunctive normal form} (DNF) if it is a disjunction of terms. We denote $\var(\varphi)$ the variables of a formula $\varphi$. Analogously, for a set of formulas $F$, $\var(F)$ denotes $\bigcup_{\varphi \in F} \var(\varphi)$. We identify finite $F$ with the conjunction of all formulas from $F$, that is $\bigwedge_{\varphi \in F} \varphi$. A mapping $\sigma: \var(\varphi) \mapsto \{0,1\}$ is called an \emph{assignment} to the variables of $\varphi$. A \emph{model} of a formula $\varphi$ is an assignment to $\var(\varphi)$ that satisfies $\varphi$. For two formulas $\psi, \varphi$ we write $\psi \models \varphi$ if every model of $\psi$ also satisfies $\varphi$. For a set of variables $V$, $|V|$ denotes the \emph{cardinality} of $V$, and $\Lits{V}$ the set of all literals formed upon $V$, that is, $\Lits{V} = V \cup \{\neg x \mid x \in V\}$. 

\paragraph{Boolean constraint languages} 
A \emph{logical relation} of arity $k \in \mathbb{N}$ is a relation $R \subseteq \{0,1\}^k$, and $\ar(R) = k$ denotes the arity.
An ($R$-)constraint $C$ is a formula $C = R(x_1, \dots, x_k)$, where $R$ is a $k$-ary logical relation, and $x_1, \dots, x_k$ are (not necessarily distinct) variables. An assignment $\sigma$ \emph{satisfies} $C$, if $(\sigma(x_1), \dots, \sigma(x_k)) \in R$. A (Boolean) \emph{constraint language} $\Gamma$ is a (possibly infinite) set of logical relations, and a \emph{$\Gamma$-formula} is a conjunction of constraints over elements from $\Gamma$. Eventually, a $\Gamma$-formula $\phi$ is $\emph{satisfied}$ by an assignment $\sigma$, if $\sigma$ simultaneously satisfies all constraints in it. In such a case $\sigma$ is also called a \emph{model of $\phi$}. We define the two constant unary Boolean relations as $\bot = \{(0)\}$ and $\top = \{(1)\}$, and the two constant 0-ary relations as $\sf{f} = \emptyset$ and $\sf{t} = \{\emptyset\}$. 

We say that a $k$-ary relation $R$ is \emph{represented} by a propositional CNF-formula $\phi$ if $\phi$ is a formula over $k$ distinct variables $x_1, \dots, x_k$ and $\phi \equiv R(x_1, \dots, x_k)$. Note such a CNF-representation exists for any logical relation $R \subseteq \{0,1\}^k$ and hence for any $R$-constraint. Any $\Gamma$-formula (a conjunction of $\Gamma$-constraints) can thus be seen as a propositional formula and admits a CNF-representation.

We note at this point that many classical fragments of propositional logic such as e.g. Horn, dualHorn, affine, k-CNF,  can be modeled via the notion of $\Gamma$-formulas. For instance, $\Gamma = \{ \{0,1\}^2 - \{(0,0)\}, \{0,1\}^2 - \{(0,1)\}, \{0,1\}^2 - \{(1,0)\}, \{0,1\}^2 - \{(1,1)\}\}$ models exactly 2-CNF formulas.

We use the shorthand $k$-CNF to denote the corresponding finite constraint language (arity bounded by $k$). $\text{IMP}$ denotes the constraint language $\{R\}$, where $R(x,y) = \{0,1\}^2 -\{(1,0)\}$. We use the shorthands CNF, Horn, DualHorn to denote the corresponding infinite constraint languages (unbounded arity).
The notation $\text{(k-)CNF}^+$ (resp. $\text{(k-)CNF}^-$ ) denotes the version where each clause is positive (resp. negative).
We use $\affine$ to denote the set of relations representable by systems of Boolean equations modulo 2, i.e., $x_1 + \ldots + x_k \equiv q \pmod 2$ for $q \in \{0,1\}$.
As representation of each relation in a constraint language we use the defining CNF-formula.

The satisfiability problem for $\Gamma$-formulas, also known as Boolean constraint satisfaction problem, is denoted $\SAT(\Gamma)$. It asks, given a $\Gamma$-formula $\phi$, whether $\phi$ admits a model. It is well-known that $\SAT(\Gamma)$ is decidable in polynomial time if $\Gamma$ is Horn, or dualHorn, or 2-CNF, or affine, or 1-valid, or 0-valid, but in all other cases is $\NP$-complete~\cite{Schaefer1978}.

\paragraph{Propositional Abduction}
An instance of the abduction problem is given by $(\KB, H, M)$, where $\KB$ denotes the knowledge base (or \emph{theory}), $H$ is the set of hypotheses, and $M$ is the set of manifestations.
In the propositional case $\KB$ is given as a (set of) propositional formula(s), and both $H$ and $M$ are sets of variables. If not noted otherwise, we denote by $V$ the total set of variables in an abduction instance, that is, $V = \var(\KB) \cup \var(H) \cup \var(M)$, and $n = |V|$ denotes its cardinality. The \emph{symmetric abduction problem}, denoted $\MABD$, asks whether there exists an $E \subseteq \Lits{H}$ such that 1) $\KB \wedge E$ is satisfiable, and 2) $\KB \wedge E \models M$. If such an $E$ exists, it is called an $\emph{explanation}$ for $M$.
If an explanation $E$ satisfies $\var(E) = H$ it is called a \emph{full explanation}, if it satisfies $E \subseteq H$, it is called a \emph{positive explanation}. If there exists an explanation $E$, it can always be extended to a full explanation (by extending $E$ according to the satisfying assignment underlying condition 1). However, the existence of an explanation does not imply the existence of a positive explanation.
The \emph{positive abduction problem}, denoted $\PMABD$, asks whether there exists a positive explanation.

\begin{example}
    Propositional abduction.

    \begin{align*}
        \KB =  \{
        & \text{Charly-lazy} \wedge \text{Charly-alone} \rightarrow \text{Charly-truant}, \\
        & \text{Dad-shopping} \rightarrow \text{Charly-alone}, \\
        & \neg \text{Buses-running} \rightarrow \text{Charly-truant} \wedge \neg \text{Dad-shopping}
        \}\\
        H   =  \{ & \text{Buses-running}, \text{Dad-shopping}, \text{Charly-lazy}\} \\
        M   =  \{& \text{Charly-truant}\}
    \end{align*}
    The set $\{\text{Charly-lazy}, \text{Charly-alone}\}$ would explain the manifestation $\text{Charly-truant}$, but is invalid, since $\text{Charly-alone}$ is no hypothesis. However, $E_1 = \{\neg \text{Buses-running}\}$ is a valid explanation. Note that $E_1$ is neither a positive nor a full explanation.
    $E_2 = \{\neg \text{Buses-running}, \neg \text{Dad-shopping}, \text{Charly-lazy}\}$ is a full explanation, and $E_3 = \{\text{Dad-shopping}, \text{Charly-lazy}\}$ is a positive explanation.
\end{example}

We define the symmetric abduction problem for $\Gamma$-formulas, denoted $\MABD(\Gamma)$, exactly as $\MABD$, but the knowledge base $\KB$ is given as a $\Gamma$-formula. $\PMABD(\Gamma)$ is defined analogously and we write $\ABD(\Gamma)$ if the specific abduction type is not important. The classical complexity of $\MABD(\Gamma)$ and $\PMABD(\Gamma)$ is fully determined for all $\Gamma$, due to \cite{NoZa2008}. An illustration of these classifications is found in Figure~\ref{fig:M-ABD} and Figure~\ref{fig:P-M-ABD}.

\paragraph{Clones, Co-clones, Post's lattice, polymorphisms}

We denote by $\Bar{x}$ the Boolean negation operation, that is, $f(x) = \neg x$. An $n$-ary {\em projection} is an operation $f$ of the form $f(x_1, \ldots, x_n) = x_i$ for some fixed $1 \leq i \leq n$.  An operation $f : \{0,1\}^k \rightarrow \{0,1\}$ is \emph{constant} if for all $\textbf{x}\in \{0,1\}^k$ it holds $f(\textbf{x}) = c$, for a $c \in \{0,1\}$. 
For a $k$-ary operation $f \colon \{0,1\}^k \to \{0,1\}$ and $X \subseteq \{0,1\}^k$ we write $f_{\mid X}$ for the function obtained by restricting the domain of $f$ to $X$.
An operation $f$ is a \emph{polymorphism} of a relation $R$ if for every $t_1, \dots, t_k \in R$ it holds that $f(t_1, \dots, t_k) \in R$, where $f$ is applied coordinate-wise. In this case $R$ is called \emph{closed} or \emph{invariant}, under $f$, ard $\inv(F)$ denotes the set of all relations invariant under every function in $F$. Dually, for a set of relations $\Gamma$, $\pol(\Gamma)$ denotes the set of all polymorphisms of $\Gamma$. Sets of the form $\pol(\Gamma)$ are known as \emph{clones}, and sets of the form $\inv(F)$ are known as \emph{co-clones}. A \emph{clone} is a set of functions closed under 1) functional composition, and 2) projections (selecting an arbitrary but fixed coordinate). For a set $B$ of (Boolean) functions, $[B]$ denotes the corresponding clone, and $B$ is called its \emph{base}. 
In the Boolean domain clones and their inclusion structure are fully determined, commonly known as \emph{Post's lattice} \cite{Post1941}.
Interestingly, the relationship between clones and co-clones constitutes a \emph{Galois connection} \cite{Lau2006}.

\begin{theorem}
    Let $\Gamma, \Gamma'$ be constraint languages. Then $\inv(\pol(\Gamma')) \subseteq \inv(\pol(\Gamma))$ if and only if $\pol(\Gamma) \subseteq \pol(\Gamma')$
\end{theorem}

Post's lattice therefore also completely describes co-clones and their inclusion structure.
A depiction of Post's lattice as co-clones is found in figure~\ref{fig:M-ABD}.
Informally explained, every vertex corresponds to a co-clone while the edges model the containment relation.
Co-clones are equivalently characterized by constraint languages (= sets of Boolean relations) closed under primitive positive first-order definitions (pp-definitions), that is, definitions of the form
$$R(x_1, \dots, x_n) := \exists y_1, \dots, y_m\, .\, R_1(\textbf{x}_1) \wedge \dots \wedge R_k(\textbf{x}_k)$$
\noindent
where each $R_i \in \Gamma \cup \{\text{Eq}\}$ and each $\textbf{x}_i$ is a tuple over $x_1, \dots, x_n, y_1, \dots, y_m$, and $\text{Eq} = \{(0,0), (1,1)\}$. The corresponding closure operator is denoted $\left<\cdot\right>$, that is, for a constraint language $\Gamma$, $\left<\Gamma\right>$ is the co-clone $\inv(\pol(\Gamma))$ and $\Gamma$ is called its \emph{base}.
Quantifier-free pp-definitions (qfpp-definitions) are defined analogously, but existential quantification is disallowed.

\paragraph{Complexity, algorithms, reductions}
We assume familiarity with basic notions in classical complexity theory (cf.~\cite{Sipser97}) and use complexity classes $\P$, $\NP$, $\coNP$, $\NP^\NP = \SigPtwo$.
Note that hardness results (e.g. NP-hardness) in this context are obtained via polynomial many-one reductions. This type of reduction is, however, not suitable to transfer exact exponential running time, and neither subexponential running time. To consider exponential time algorithms, a \emph{complexity parameter} needs to be specified. For this paper, we consider as parameter the \emph{total number of variables} in an instance (e.g. an abduction or SAT instance). 
For a variable problem $A$ we denote $I(A)$ the set of instances and $\var(I)$ denotes the set of variables. If clear from the context we usually write $n$ rather than $|\var(I)|$.
Thus, formally, we work in the setting of {\em parameterized} complexity where the complexity parameter $k$ is $n$, the number of variables in the knowledge base in an abduction problem, or the number of vertices in a graph problem. We define the following two types of reductions~\cite{jonsson2017}.

\begin{definition}
Be $A,B$ two variable problems. A function $f \colon I(A) \rightarrow I(B)$ is a \emph{many-one linear variable reduction} (LV-reduction) with parameter $C \geq 0$ if:
\begin{enumerate}
\item $I$ is a yes-instance of $A$
iff $f(I)$ is a yes-instance of $B$,
\item $|\var(f(I))| = C \cdot |\var(I)|+O(1)$, and 
\item $f$ can be computed in polynomial time.
\end{enumerate}
\end{definition}

If in an LV-reduction the parameter $C$ is $1$, we speak of a \emph{CV-reduction}, and we take us the liberty to view a reduction which actually shrinks the number of variables ($|\var(f(I))| \leq |\var(I)| +O(1)$) as a CV-reduction, too. We use $A \cvred B$, respectively $A \lvred B$ as shorthands, and sometimes write $A \cveq B$ if $A \cvred B$ and $B \cvred A$.
For  algorithms' (exponential) running time and space usage we adopt the $O^*$ notation which suppresses polynomial factors. 
A CV-reduction transfers exact exponential running time: if $A \cvred B$, and $B$ can be solved in time $O^*(c^n)$, then also $A$ can be solved in time $O^*(c^n)$ (where $n$ denotes the complexity parameter). LV-reductions, on the other hand, preserve {\em subexponential complexity}, i.e., if $B$ can be solved in $O^*(c^n)$ time for every $c > 1$ and $A \lvred B$ then $A$ is solvable in $O^*(c^n)$ time for every $c > 1$, too. For additional details we refer the reader to~\cite{tcs2021}.

\section{Missing Proofs}

\subsection{Proof of Theorem 3}

\begin{proof}
Let $(\KB, H, M)$ be an instance of $\MABD(\Gamma)$. We first observe that
we without loss of generality may assume that $\var(E) = H$. We then exhaustively enumerate all full explanations $E \subseteq \Lits{H}$ ($2^{|H|}$ time since we do not have to consider $E$ containing both a variable and its complement) and use the algorithm for $\SAT(\Gamma^+)$ to 
\begin{enumerate}
\item check whether $\KB \wedge E$ is satisfiable, and 
\item check whether $\KB \land E \land \neg m_i$ is unsatisfiable for each $m_i \in M$.
\end{enumerate}
Since $\var(E) = H$, each instance $\KB \wedge E$ effectively only has $n - |H|$ variables, we therefore get the bound $2^{|H|} \cdot f(n - |H|) \cdot (|M| + 1)$. Note that checking satisfiability of $\KB \wedge E$ can be done by the presumed algorithm for $\SAT(\Gamma^+)$ by forcing variables in $E$ constant values via $\bot$- and $\top$-constraints.

For $\PMABD(\Gamma)$ the analysis is similar but where in each invocation of the algorithm for $\SAT(\Gamma^+)$ we have $n - |E|$ variables.
\end{proof}

\subsection{Proof of Lemma 6}

\begin{proof}
Assume there is an $m \in M - \var(\KB)$.
If $m \notin H$, there are no explanations. An algorithm returns False and terminates, a reduction maps to a fixed negative instance.
If $m \in H$, then $m$ can trivially be explained (by itself). It can thus be removed from both $H$ and $M$, since it does not influence the existence of explanations.

Now assume there is an $h \in H - \var(\KB)$.
If $h \notin M$, it can be removed from $H$, since it does not influence the existence of explanations.
If $h \in M$, it can be removed from both $H$ and $M$, since it does not influence the existence of explanations.

Since no variables are added, the described transformations constitute a CV-reduction.
\end{proof}

\subsection{Proof of Theorem 7}

\begin{proof}
Consider an instance $(\KB, H, M)$ of $\MABD(\Gamma)$. Recall that by Lemma~6
we can W.L.O.G. assume that $M,H \subseteq \var(\KB)$. Let $|\var(\KB)| = n$. 

    The algorithm is presented in Algorithm~\ref{alg:enumABD} where we assume that $\mathrm{modelGenerator}(\KB)$ can be used to enumerate all solutions to $\KB$ via the $.\mathrm{next}()$ notation. 
    For soundness, we first recall that there exists an explanation $E \subseteq \Lits{H}$ if and only if there exists full explanation. The algorithm enumerates all models $\sigma$ of $\KB$ and in line 6 extracts a full explanation candidate $E$ from $\sigma$. This $E$ serves as identifier of $\sigma$'s equivalence class. An equivalence class is discarded if it fails (via the representative $\sigma$) to explain $M$. This case is detected in line 8.
    For the complexity, we assume that all models of $\SAT(\Gamma)$ can be enumerated in $O^*(c^n)$ time, which bounds the total number of iterations. Each step \emph{after} line 5, that is, lines 6-12,  can be done in polynomial time, and all checks can be implemented with standard data structures. Note that the check in line 8 amounts to simply evaluating $M$ on the assignment $\sigma$.
\end{proof}

\begin{algorithm}[H]
\caption{Algorithm to solve $\MABD(\Gamma)$ based on enumeration.}
\label{alg:enumABD}
	\begin{algorithmic}[1]
		\REQUIRE $\KB, H, M$
			\STATE generator = modelGenerator($\KB$)
			\STATE discarded = $\emptyset$
			\STATE potentialExp = $\emptyset$
			\WHILE{generator.notEmpty()}
  			\STATE $\sigma$ = generator.next()
  			\STATE $E = \{x \in H \mid \sigma(x) = 1\} \cup \{\neg x \mid x \in H, \sigma(x) = 0\}$
				\IF{$E \notin$ discarded}
  				\IF{$\sigma \not \models M$}
						\STATE discarded = discarded$\, \cup\, \{E\}$
						\STATE potentialExp = potentialExp$\, -\, \{E\}$					\ELSE 
					  \STATE potentialExp = potentialExp$\,\cup\, \{E\}$
					\ENDIF
				\ENDIF
			\ENDWHILE
		    \STATE \emph{\# now potentialExp contains all full explanations}
			\IF{potentialExp $\neq \emptyset$}
				\STATE return $\top$
			\ELSE
				\STATE return $\bot$
			\ENDIF
		\end{algorithmic}
\end{algorithm}

\subsection{Proof of Theorem 8}

\begin{algorithm} 
    \caption{Algorithm $\mathcal{A}$ for $\PMABD(\Gamma)$.}
    \label{alg:PABD}
    \begin{algorithmic}[1]
        \REQUIRE $\KB, H,M,D,\delta$
                \STATE $E = D \cup \delta$
				\STATE $G = E \cup \{\neg x \mid x\in H-E\}$
        \IF{$\KB \wedge G \wedge \neg M$ is satisfiable}
           \RETURN $\bot$
        \ENDIF
        \IF{$\KB \wedge G$ is satisfiable}
            \RETURN $\top$
        \ELSE
           \IF{$D = \emptyset$}
             \RETURN $\bot$
           \ENDIF
           \STATE flag = $\bot$
           \FOR{$x \in D$}
             \STATE flag = flag $\vee$ $\mathcal{A}(\KB,H,M,D - \{x\},\delta)$
             \STATE $\delta = \delta \cup \{x\}$
           \ENDFOR
            \RETURN flag
        \ENDIF
				
    \end{algorithmic}
\end{algorithm}

\begin{proof}
Consider Algorithm~\ref{alg:PABD}. The starting parameters are $D = H$ and $\delta = \emptyset$.

The recursive algorithm systematically explores all subsets of $H$ as candidates.
It starts with the base candidate $E = H$. Inside each recursive call, it first checks if the current candidate $E$, extended to a full candidate $G$, entails the manifestation (line 3). If this fails, it concludes that neither $G$ nor any subset of $G$ (which includes $E$ and all its subsets) can be an explanation, and returns $\bot$. Else, it checks if $G$ is consistent with $\KB$. If yes, then $G$ is obviously a (full) explanation, but the algorithm concludes that in this case even $E \subseteq G$ is a (positive) explanation. In claim~\ref{claim1}  below we argue why this is correct. If $G$ is not consistent with $\KB$, the algorithm concludes that neither $E$ is consistent with $\KB$, and can thus not be an explanation. In claim~\ref{claim2} below we argue why this is correct.
Then the algorithm systematically checks candidates where exactly one variable is removed from $E$ (lines 11-13). Descending in the recursive calls (line 12) it makes sure to systematically explore \emph{all} subsets of a candidate, thereby avoiding visiting the same subset multiple times (this is the purpose of $\delta$).

\begin{claim}\label{claim2}
    If in line 5 of the algorithm $\KB \wedge G$ is unsatisfiable, then so is $\KB \wedge E$.
\end{claim}
\begin{proof}
    We use an inductive argument. At the initial call we have that $E = G = H$, so the statement is obviously true. Now, consider the case where $E$ is a strict subset of $H$. By construction of the algorithm we know that any strict superset $F$ of $E$ must have failed as explanation candidate (otherwise no recursive call would have descended to the current point). Since a failure due to line 3 stops descending, the failure (of $F$) must in this case have been due to line 5. We state this as the following observation.
    \begin{equation}\label{obs1}
      \text{If } F \supset E, \text{ then } \KB \wedge F \text{ is unsatisfiable.}
    \end{equation}
    Now assume for contradiction that $\KB \wedge E$ is satisfiable. Then there must be a witnessing satisfying assignment $\sigma$. We define $C = \{x \mid x\in H-E, \sigma(x) = 1\}$ and conclude that $\KB \wedge E \wedge C$ must be satisfiable (by $\sigma$). If $C$ is non-empty, then $E \cup C$ is a strict superset of $E$, thus by  observation~\ref{obs1}, $\KB \wedge E \wedge C$ must be unsatisfiable, a contradiction. If $C$ is empty, $\sigma$ assigns 0 to all variables in $H - E$. Therefore, $\sigma$ satisfies $\KB \wedge G$, a contradiction.

\end{proof}

\begin{claim}\label{claim1}
    If in line 3 of the algorithm $\KB \wedge G \wedge \neg M$ is unsatisfiable, then so is $\KB \wedge E \wedge \neg M$.
\end{claim}
\begin{proof}
As in claim~\ref{claim2} we obtain the same observation~\ref{obs1}. The remaining reasoning proceeds exactly analogously, only with $\KB \wedge \neg M$, instead of $\KB$.
\end{proof}

It remains to observe the claimed time and space complexity. In the worst case the algorithm goes through all subsets of $H$, and checks satisfiability in the remaining $V-H$ variables. Its time complexity is thus $O^*(2^{|H|}) \cdot O^*(2^{|V-H|}) = O^*(2^n)$. The descending depth is linear and no set can grow to exponential size. The overall space usage is therefore polynomial.
\end{proof}

\subsection{Proof of Theorem 9}

\begin{proof}
Similar to the enumeration scheme of Algorithm~\ref{alg:enumABD}, we use a generator with the $.\mathrm{next}()$ notation to enumerate all models of the knowledge base $\KB$. But this time we require a specific order. For an assignment $\sigma$ we define the \emph{weight with respect to $H$} as $w_H(\sigma) = |\{x \mid x \in H, \sigma(x) = 1\}|$. We require the generator to respect the order of non-increasing $w_H(\sigma)$.
Algorithm~\ref{alg:enumPABD} assumes that $L$ is a given list of all models of $\KB$ in this order, and that the generator created in line 1 follows this order. Note that $L$ can be constructed in time and space $O^*(c^n)$ by generating all models with the given enumeration algorithm in time $O^*(c^n)$, storing them, and then sorting them in time $O^*(c^n \cdot log(c^n)) = O^*(c^n)$.

The effect of this order is that the algorithm will encounter \emph{subset-maximal candidates} first.
A positive explanation candidate $E \subseteq H$ is called \emph{subset-maximal candidate} if $\KB \wedge E$ is satisfiable and $E$ is subset-maximal with this property. 
It is readily observed that to decide whether there are any positive explanations, it is sufficient to check all subset-maximal candidates. As we explain in the following, this is exactly what the algorithm does.

After line 1 the algorithm proceeds exactly as Algorithm~\ref{alg:enumABD}, except that in line 6 a \emph{positive} explanation is extracted (instead of a \emph{full} explanation) and that in line 13-15 it is made sure that if a candidate $E$ is discarded, also all (immediate) subsets are discarded.
Note that if a subset-maximal candidate $E$ fails to explain $M$ (detected in line 8), then also any strict subset of $E$ will fail.
It is thus \emph{legitimate} to discard all subsets of a discarded $E$.
But it is even \emph{necessary} for the algorithm's correctness: it assures that the algorithm continues to only consider subset-maximal candidates.
\emph{If} the algorithm \emph{did} hit a non-subset-maximal candidate $E$ (that is not listed in the discarded set), it could miss a witness of $E$ failing to explain $M$. That is, it could miss a model $\sigma$ of $\KB \land E$ such that $\sigma \not\models M$, since this witness might already have been 'consumed' with a superset of $E$.

\end{proof}

\begin{algorithm}[H]
\caption{Enumeration based algorithm for $\PMABD(\Gamma)$.}
\label{alg:enumPABD}
	\begin{algorithmic}[1]
		\REQUIRE $L, H, M$
			\STATE generator = Generator($L$)
			\STATE discarded = $\emptyset$
			\STATE potentialExp = $\emptyset$
			\WHILE{generator.notEmpty()}
  			\STATE $\sigma$ = generator.next()
  			\STATE $E = \{x \in H \mid \sigma(x) = 1\}$
				\IF{$E \notin$ discarded}
  				\IF{$\sigma \not \models M$}
						\STATE discarded = discarded$\, \cup\, \{E\}$
						\STATE potentialExp = potentialExp$\, -\, \{E\}$					\ELSE 
					  \STATE potentialExp = potentialExp$\,\cup\, \{E\}$
					\ENDIF
				\ENDIF
                \IF{$E \in$ discarded}
                \STATE \emph{\# discard all immediate subsets}
                \STATE discarded = discarded$\, \cup\, \{E - \{x\} \mid x \in E\}$
                \ENDIF
			\ENDWHILE
		    \STATE \emph{\# now potentialExp contains all subset-maximal positive explanations}
			\IF{potentialExp $\neq \emptyset$}
				\STATE return $\top$
			\ELSE
				\STATE return $\bot$
			\ENDIF
		\end{algorithmic}
\end{algorithm}

\subsection{Proof of Theorem 15}

\begin{proof}
  Let $r_0$ be the integer such that for every $r \geq r_0$ it holds
  that $|R| \leq c^{r}$ for every $r$-ary $R \in \Gamma$ (for a fixed $c < 2$).
  Such $r_0$ exists since $\Gamma$ is asymptotically sparse. We first explain how to exhaustively branch until we reach sufficiently low arity $< r_0$.

  Now let $(V,C)$ be an instance of $\SAT(\Gamma)$ with maximum arity $r \geq r_0$. For a partial truth assignment $f \colon X \to \{0,1\}$, $X \subseteq V$, we write $C_{\mid f} = \{c_{\mid f} \mid c \in C\}$ where $c = R(\mathbf{x})$ and $c_{\mid f} = R_{\mid f}(\mathbf{x}')$ with the scope $\mathbf{x'}$ obtained by removing every variable in $X$ from $\mathbf{x}$. Each such $R_{\mid f} \in \Gamma$ since we assume that $\Gamma$ is closed under branching.
  
  Let $A$ be the following
  branching algorithm, which, given a set of constraints $C$, for
  every constraint $R(x_1, \ldots, x_k) \in C$, $\ar(R) = k$, $r_0 \leq k \leq r$,
  branches as follows:
  \begin{itemize}
  \item
  $A((C \setminus \{R(x_1, \ldots, x_k)\})_{\mid f_1})$,
  \item
    \vdots
      $A((C \setminus \{R(x_1, \ldots, x_k)\})_{\mid f_{|R|}})$,
  \end{itemize}
  where $f_i$ is the partial truth assignment corresponding to $t_i \in
  R$. Since $\Gamma$ is closed under minors we may
  assume that the variables in the constraint $R(x_1, \ldots, x_k)$ are
  all distinct.
  Hence, in every branch we eliminate $k$ variables and in
  total there are $|R| \leq c^k$ branches. In the worst case this
  gives a bound of $|R|^{\frac{n}{k}} \leq c^{k^{\frac{n}{k}}} \leq c^n$. 

      
Note that if we apply branching again to a constraint $R'_{\mid f_i}(y_1, \ldots, y_{k'})$ for some $1 \leq i \leq |R|$, then we may again assume that all variables $y_1, \ldots, y_{k'}$ are distinct (since $\Gamma$ is closed under minors), and, importantly, that $|R| \leq c^{k'}$, and we may simply apply the branching algorithm $A$ again to get the required bound.

We keep this up until we in a branch reach an instance $C_{r_0}$ where every constraint has arity $< r_0$. Pick a constraint $R''(z_1, \ldots, z_{k''})$ where $k'' < r_0$. We observe that $|R''| < 2^{k''} < 2^{r_0}$ where the first equation holds via the assumption that $R''$ is non-trivial, and the worst possible bound is then smaller than $(2^{r_0} - 1)^{\frac{n}{r_0}} = (2^{r_0} - 1)^{\frac{1}{r_0}^{n}} = d^n < 2^n$ since $r_0 \geq 1$ is constant. Put together, the total number of models is dominated by either $c^n$ or $d^n$, both better than $2^n$.
\end{proof}

\subsection{Proof of Lemma 17}

\begin{proof}
    First, let $R_S$ be a $k$-ary relation induced by $p,q \leq k + 1$, i.e., the sum of all arguments is congruent to $q$ modulo $p$. Clearly, inserting 0 or 1 in the equation simply mean that we have to adjust $p,q$, and $k$, as necessary. Similarly, $R_S \subset \{0,1\}^{k}$ since $S \subset [k] \cup \{0\}$ (each equation has at least one weight which is not allowed).

    Second, we only have to establish that all relations in $\minor(\mathcal{E}^{\leq k})$ are non-trivial and closed under substitutions. We provide a short proof sketch since it follows from basic properties of equations. Thus, assume that we are given a minor of an equation in $\mathcal{E}^{\leq k}$, and then assign one or more variables constant values. Then the resulting equation can be obtained from the original equation by first assigning the corresponding variables constant values, and then taking a minor of the resulting relation. To additionally see that every relation is non-trivial, we simply observe that an equation of the form $x_1 + \ldots + x_{\ell} \equiv q \pmod p$, $\ell \leq k$, cannot become trivially satisfied by identifying variables.
\end{proof}

\subsection{Proof of Theorem 18}

\begin{proof}
We begin by proving that $\XSAT$ satisfies the necessary assumptions. First, consider a relation $R \in \XSAT$ of arity $k$. Then either $R = R_{1/k}$ in which case $|R_{1/k}| = k$, or $R = \bot^k$ in which case $|\bot^k| = 1$, so the language is clearly asymptotically sparse. Second, for an assignment $f \colon X \to \{0,1\}$, $X \subseteq [k]$, for which $R_{\mid f} \neq \emptyset$, note that if $f(i) = 0$ for every $i \in X$, then $R_{\mid f} = R_{1/k'}$ for $k' = k - |X|$. Similarly, if $f(i) = 1$ for some $i \in X$, then $R_{\mid f} = \bot^{k'}$ for some $k' = k - |X|$. Relations of the form $\bot^k_{\mid f}$ are also simple to handle. Third, let $g \colon [k] \to [k']$, and consider the minor defined by $R_{1/k}(x_{g(1)}, \ldots, x_{g(m)})$. A simple case analysis shows that if $g(i) = g(j)$ for $i \neq j$ 
To then obtain the concrete bound $O^*(\sqrt{2}^n)$ we simply observe that when branching on a $k$-ary constraint we obtain a branching vector of \[(\overbrace{k, k, \ldots, k}^{\text{$k$ branches}})
\] and by solving the resulting recurrence equation it is easy to see that the worst possible case is when $k = 2$, giving the bound $O^*(\sqrt{2}^n)$.

Second, for $\affine^{\leq k}$, we observe that it is asymptotically sparse since it is finite, every relation is non-trivial since $\affine^{\leq k} \subseteq \textsc{Equations}$, and that it is closed under branching follows from basic properties of affine equations. However, in contrast to $\textsc{Equations}$, it is closed under minors: an equation is true if and only if the total parity is even or odd, and identifying two or more variables does not affect this relationship.
\end{proof}

\subsection{Proof of Theorem 21}

\begin{proof}
Let $(\KB, H, M)$ be an instance of $\MABD(\text{k-CNF}^+)$.
Note that there exists an explanation if and only if there exists a purely negative explanation.
Denote by $\cal{C}$ the set of clauses in $\KB$, that is, $\KB = \bigwedge_{C \in \cal C} C$, where each clause $C$ is positive and of size $k$. One further verifies that to determine whether there are any explanations, one only needs to consider two types of clauses: (1) clauses that contain only variables from $H$ (seen as a "constraint" on how to select variables from $H$ as negative literals), and (2) clauses that contain exactly one variable $m \in M$ and  variables from $H$ (that is, $\neg h_1, \dots, \neg h_{k-1}$ is a candidate to explain $m$, via the clause $(h_1 \vee \dots \vee h_{k-1} \vee m)$). Therefore, we define the two sets of clauses 

\begin{align*}
 F =\; & \{(h_1 \vee \dots \vee h_{k-1} \vee m) \in {\cal C} \mid h_1, \dots, h_{k-1} \in H, m \in M\}, \\
 G =\; & \{(h_1 \vee \dots \vee h_k) \in {\cal C} \mid h_1, \dots, h_k\in H\}, \\
\end{align*}

\noindent
and disregard henceforth all other clauses.
For each $m \in M$ define the formula
 $$D_m = \bigvee_{(h_1 \vee \dots \vee h_{k-1} \vee m) \in F}  (\neg h_1 \wedge \dots \wedge \neg h_{k-1})$$
Note that $D_m$ is a negative DNF formula, that is, a disjunction of negative terms.
We map $(\KB, H, M)$ to the $\simpleSAT^k$ instance 

$$\varphi = \bigwedge_{(h_1 \vee \dots \vee h_k) \in G} (h_1 \vee \dots \vee h_k) \wedge \bigwedge_{m\in M} D_m$$

\noindent
Note that $\varphi$ is of required form and that $\var(\varphi) \subseteq H$.
It is easy to verify that there is a one-to-one correspondence between negative explanations and models of $\varphi$.
\end{proof}

\subsection{Proof of Lemma 22}

\begin{proof}
    Consider the instance $(\KB, H, M)$ where $\KB$ contains clauses of type $(\neg x_1 \vee \dots \vee \neg x_k)$ and $(x \rightarrow y)$.
    Recall that by Lemma~6 we may W.L.O.G assume that $M,H \subseteq \var(\KB)$. Note further that we also may assume that $H \cap M = \emptyset$: define $X = H \cap M$, remove $X$ from both $H$ and $M$, introduce fresh variables $h,m$, add $h$ to $H$, add $m$ to $M$, and add to $\KB$ the clauses $(h\rightarrow m)$ and $\bigwedge_{x \in X}(h\rightarrow x)$. This constitues a CV-reduction of the problem to itself.
    
    For a variable $h\in H$ denote by $\cons(h)$ the consequences of $h$ in $\KB$, that is, all single literals that can be derived by applying exhaustively resolution on $\KB \wedge h$.
    We map $(\KB, H, M)$ to $(\KB', H, M)$, where
    $\KB' = \varphi_1 \wedge \varphi_2$, with

    \begin{align*}
        \varphi_1 = & \bigwedge_{\begin{array}{c}m\in M, h\in H \\ \text{s.t.}\, m \in \cons(h)\end{array}} (h \vee m)\\ 
        \varphi_2 = & \bigwedge_{\begin{array}{c}h_1, \dots, h_k \in H\;\text{s.t.}\\  \KB \wedge M \wedge h_1 \wedge \dots \wedge h_k \models \emptyset\end{array}} (h_1 \vee \dots \vee h_k) \\
    \end{align*}

\noindent
It is easy to verify that there is a one-to-one correspondence between explanations, obtained by flipping the literals of an explanation.
\end{proof}

\subsection{Proof of Theorem 26}

\begin{proof}
We consider the problem of deciding the truth value of a quantified Boolean formula $\forall X \exists Y . \Phi(X,Y)$ where $\Phi$ is in 3-CNF, i.e., whether there for all possible assignments to the set of variables $X$ exists an assignment to the set of variables $Y$ such that $\Phi(X,Y)$ under those assignments is true. Then, from \cite{calabro2013exact} we know that this problem can not be solved in $c^n$ time for any $c < 2$ under SETH. 

the complement of this problem is $\exists X \forall Y \Phi' (X,Y)$ where $\Phi' = \neg \Phi$ is in 3-DNF can also not be done in time less than $2^n$ under SETH. 

It suffices to have a CV reduction from $\forall X \exists Y \Phi (X,Y)$ where $\Phi$ is in 3-DNF to $\MABD(\text{4-CNF})$, inspired by the reduction from \cite{eiter1995complexity}.

Let $\exists x_1,x_2 \ldots \forall y_1,y_2 \ldots \Phi$ be in 3-$\textsc{DNF}$ form. Define the following $\MABD(4\textsc{-CNF})$ instance: $(\KB, H, M)$ where $H= \{x_i \}$, $M=\{y_i\} \cup \{s\}$, and $\KB= \{\Phi \rightarrow s \land y_1 \land y_2 \ldots\} \land \{s \rightarrow  y_1 \land y_2 \ldots\}$. First, we show that $\KB$ can be expressed in 4-CNF. Her, $\Phi$ is in 3-DNF, $\{\Phi \rightarrow s \land y_1 \ldots\}$, and can be written as $\{\neg \Phi \vee \{s \land y_1 \ldots\}\}$, $\{\neg \Phi\}$, which is in 3-CNF, and distributing the $\vee$ makes 4-CNF clauses. The second part $\{s \rightarrow  y_1 \land y_2 \ldots\}$ follows the same reasoning to make it into 2-CNF.

We prove the correctness of our reduction. From left to right: assume that the QBF formula is true. Then there is $X \subseteq H$ such that for all $Y$, $\Phi$ is true. The theory $\KB$ is satisfiable for $X$ by putting all $y_i$ and $s$ to true. The entailment of $Y$ is easy to see, as putting any $y_i$ or $s$ to $\bot$ makes the theory evaluate to false if $\phi$ is true.

Now, from right to left: we assume the abduction problem does have a solution, and show that the QBF formula is true. Let $X \subseteq H$ be a solution to the abduction problem. $X$ satisfies $\phi$ when all $y_i$ and $s$ are true, since $\KB$ must be evaluated to true when the both the manifestations and solutions of the abduction problem are positive. Now let us assume $s$ is false. By entailment $\KB$ must be unsatisfiable, meaning that $\{\Phi \rightarrow s \vee y_i \ldots\}$ must be false. This can only be the case if $\Phi$ is $\top$, and this applies for all $Y$, thus completing the proof.
\end{proof}

\subsection{Proof of Lemma 27}

\begin{proof}
We reduce a regular $\MABD(\text{4-CNF})$ instance into a $\PMABD(\text{4-CNF})$ as follows:
$(\KB,H,M) \rightarrow (\KB',H',M)$ where $\KB' = \KB \cup \{x \leftrightarrow \neg x' \mid x \in H\}$ and $H' = H \cup \{x'\}$.

It is easy to see that the second problem has a positive solution if and only if the first one has any solution. The inequalities between the $x$ and $x'$ variables can be written in 2-CNF, thus $\KB'$ is in 4-CNF. We notice that the number of variables has increased linearly, but remains less than twice the number of variables of the original problem, thus, $\PMABD(\text{4-CNF})$ cannot be solved faster than $2^{n/2} = 1.4142^{n}$
\end{proof}

\subsection{Proof of Theorem 28}
\begin{proof}
We present the proof for $\PMABD(\Gamma)$ but the construction is exactly the same for $\MABD(\Gamma)$. Note that $\PMABD(\Gamma) \cvred \PMABD(\Gamma \cup \{\bot, \top\})$ trivially holds, so we only need to prove the other direction.
For this, we first claim that $R_{\neq} = \{(0,1), (1,0)\}$ is qfpp-definable by $\Gamma$. Choose a relation $R \in \Gamma$ of arity $k$ such that $\emptyset \subset R \subset \{0,1\}^{k}$ and such that $R$ does not contain either of the two constant tuples $(0, \ldots, 0)$ and $(1, \ldots, 1)$. Such a relation must exist since otherwise $\pol(\Gamma) \supset [\Bar{x}]$ (note that $(0, \ldots, 0) \in R$ but $(1, \ldots, 1) \notin R$ is not possible since $\overline{(0, \ldots, 0)} = (1, \ldots, 1)$). Second, choose a non-constant tuple $t = (a_1, \ldots, a_k) \in R$, and construct two sets of indices $I_0$ and $I_1$ such that $i \in I_b$ if and only if $a_i = b$ for $b \in \{0,1\}$. If we then define the relation $S(x,y) \equiv R(x_{1}, \ldots, x_k)$ where $x_i = x$ if $i \in I_0$, and $x_i = y$ if $i \in I_1$ it follows (1) that $(0,1) \in S$ due to the assumptions on the tuple $t$, and (2) that $\Bar(t) = (1,0) \in S$ since $\pol(\Gamma) = [\Bar{x}]$. Finally, if $(0,0) \in S$ or $(1,1) \in S$ then $(0, \ldots, 0) \in R$ or $(1, \ldots, 1) \in R$, contradicting our original assumption. We conclude that $S = R_{\neq}$ and for simplicity assume that $R_{\neq} \in \Gamma$ to simplify the notation in the remaining proof.

Second, let $(\phi, H, M)$ be an instance of $\PMABD(\Gamma \cup \{\bot, \top\})$. Partition $\phi$ into $\phi_1 \cup \phi_2$ where $\phi_1$ contains only constraints from $\Gamma$, and $\phi_2$ only constraints from $\{\bot, \top\}$. We introduce two fresh variables $V_0$ and $V_1$ and create the instance $(\phi', H', M')$ of $\PMABD(\Gamma)$ where:
\begin{enumerate}
    \item 
    $\phi' = \phi_1 \land \bigwedge_{\bot(x) \in \phi_2} R_{\neq}(x, V_1) \land \bigwedge_{\top(y) \in \phi_2} R_{\neq}(y, V_0) \land R_{\neq}(V_0,  V_1)$,
    \item 
    $H' = H \cup \{V_1\}$, and
    \item 
    $M' = M \cup \{V_1\}$.
\end{enumerate}

Assume that $E \subseteq H$ is an explanation for $(\phi, H, M)$. Then $\phi' \land (E \cup \{V_1\})$ is satisfiable and explains $M \cup \{V_1\}$. 
For the other direction, assume that $E' \subseteq H'$ is an explanation for $(\phi', H', M')$. Since $V_1 \in M'$ we either have $V_1 \in E$, or $y \in E$, for some $\top(y) \in \phi_2$, which via the constraints $\bigwedge_{\bot(x) \in \phi_2} R_{\neq}(x, V_1)$ and $R_{\neq}(V_0, V_1)$ ensures that $E' \cap H$ is an explanation for $(\phi, H, M)$.
\end{proof}

\subsection{Proof of Theorem 29}

We need a little bit of additional notation for the next theorem. Let $s = (s_1, \ldots, s_k) \in [0,1]^k$ be a $k$-ary {\em sign-pattern}, let $0_s = (a_1, \ldots, a_k) \in \{0,1\}^k$ be the unique tuple such that the sum $(s_1 \oplus a_1) + \ldots (s_k \oplus a_k) = 0$, and $1_s$ the corresponding tuple whose sum is $k$. Finally, let $R^s_{\mathrm{NAE}} = \{0,1\}^k \setminus \{0_s, 1_s\}$. It is easy to see that one for every $k$-clause $(\ell_1 \lor \ldots \lor \ell_k)$ can construct a sign-pattern and thus a corresponding $R^s_{\mathrm{NAE}}$ relation. We let $k$-NAE be the set of all such relations of arity $k$.

\begin{proof}
Again, we present the proof for $\PMABD(\Gamma)$ but the construction is exactly the same for $\MABD(\Gamma)$. Let $(\phi, H, M)$ be an instance of $\PMABD(k-CNF)$. We introduce two fresh variables $V_0, V_1$, the constraint $R_{\neq}(V_0, V_1)$ (recall the proof of Theorem 28, and construct $\phi' = \{R^s_{\mathrm{NAE}}(x_1, \ldots, x_k, V_0) \mid (\ell_1 \lor \ldots \ell_k) \in \phi\}$ where the $i$th component of the sign-pattern $s$ is $0$ if $\ell_i$ is positive, and equal to $1$ otherwise. Let $E' = E \cup \{V_1\}$ and $M' = M \cup \{V_1\}$. The correctness then follows via similar arguments as for satisfiability (see, e.g., Lemma 47 in \cite{jonsson2017}).
\end{proof}

\subsection{Proof of Lemma 30}

\begin{proof}
We give a reduction from CNF-SAT which at most triplicates the number of variables ($2^\frac{1}{3} > 1.2599$). The second bound follows since the reduction satisfies $|H| = 2\cdot |\var(\varphi)|$ (and $2^\frac{1}{2} > 1.4142$). The last two bounds follow since in the reduction $|M| = |\var(\varphi)|$ and $\frac{|H|}{|M|} = 2$.

Be $\varphi = \bigwedge_{j\in J} C_j$ an instance of CNF-SAT. For each variable $x \in\var(\varphi)$ introduce two fresh variables, denoted $x'$ and $m_x$. Denote by $C_j'$ the clause obtained from $C_j$ where each positive literal $x$ is replaced by $\neg x'$. That is, $C_j'$ contains negative literals only. Define the abduction instance as $(\KB, H, M)$, where

\begin{align*}
    \KB  = & \bigwedge_{j\in J} C_j' \\
        & \wedge \bigwedge_{ x \in \var(\varphi)}(\neg x \vee \neg x') \wedge (x \rightarrow m_x) \wedge (x' \rightarrow m_x) \\
    H   = & \{x, x' \mid x \in \var(\varphi)\} \\
    M   = & \{m_x \mid x  \in \var(\varphi)\}
\end{align*}
\end{proof}

\subsection{Proof of Lemma 32}
\begin{proof}
    From~\cite{lance}, we know that CNF-SAT does not admit a kernel of size $f(n)$ where $f$ is a polynomial, i.e., an equisatisfiable instance with e.g.\ a quadratic number of clauses. To prove the claim it suffices to have a reduction from $\MABD(\text{2-CNF})$ to CNF-SAT that only has a polynomial number of clauses in terms of $n$. Assume an arbitrary $\MABD(\text{2-CNF})$ instance. For every $m_i \in M$, for all clauses of the type: $(\ell_j\lor m_i), \ldots, (\ell_k \lor m_i)$ where $\ell_i \ldots \ell_k$ are a literals, we delete these clauses and add the following CNF clause $(\ell_i \lor \ldots \lor \ell_k)$, and otherwise keep everything unchanged. This completes the reduction. The correctness is straightforward: for every manifestation, we have to choose at least one of the literals that appear with it in a clause, in order to entail it. The $\MABD(\text{2-CNF})$ instance can only have up to $n^2$ clauses. The resulting CNF-SAT instance from this reduction thus only has a number of clauses equal to the number of $m_i \in M$, plus the remaining 2-CNF clauses, thus  at most $n^2$ clauses. If an LV-reduction of the form CNF-SAT $\lvred \MABD(\text{2-CNF})$ exists, together with the reduction from this proof, we would have a reduction from an arbitrary CNF-SAT instance into a CNF-SAT with only $n^2$ clauses, which is not possible unless $\NP \subseteq P/Poly$.
\end{proof}

\newpage

\begin{figure*}[htp]
    \centering
    \includegraphics[width=14cm]{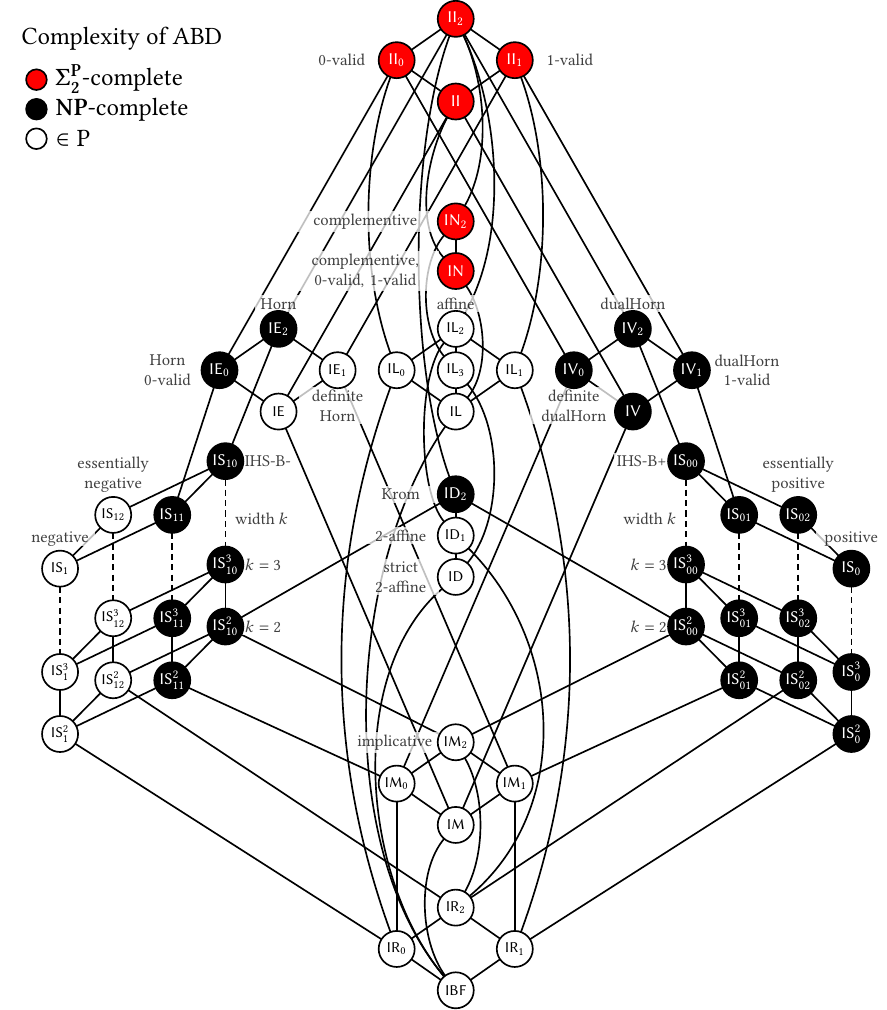}
    \caption{Classical complexity of ABD according to Nordh and Zanuttini (2008), illustrated on Post's lattice.}
    \label{fig:M-ABD}
\end{figure*}

\newpage

\begin{figure*}[htp]
    \centering
    \includegraphics[width=14cm]{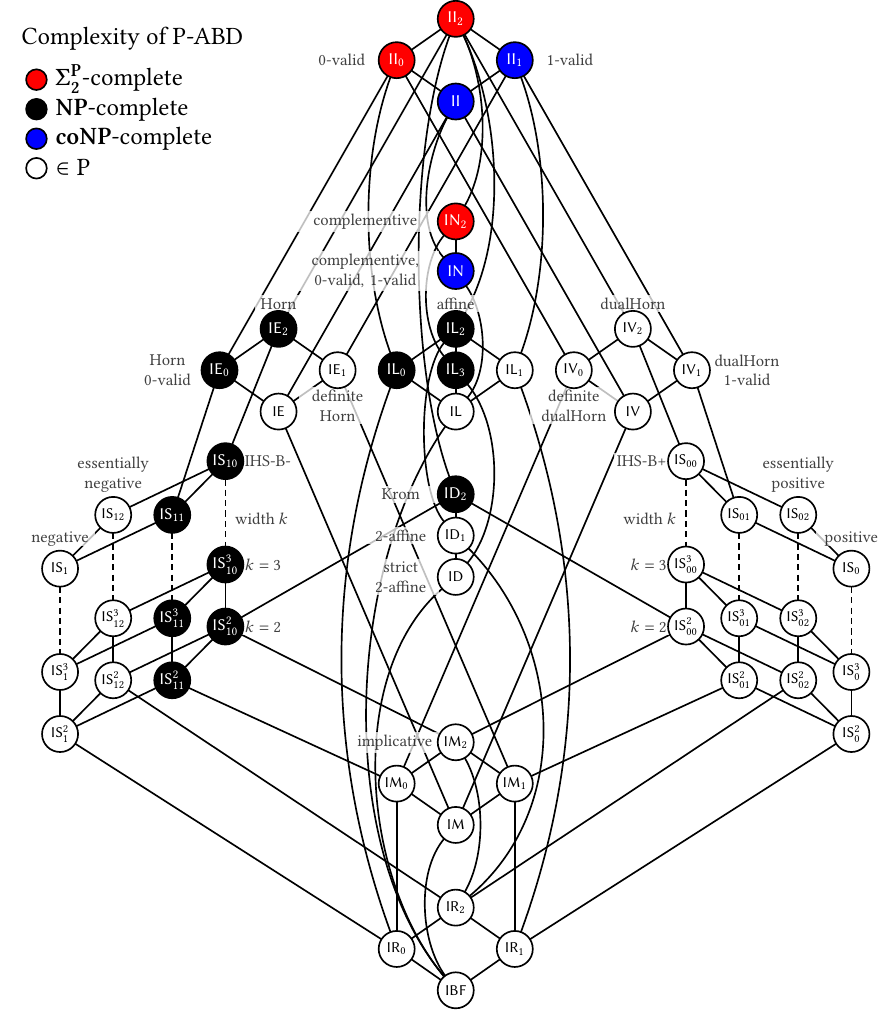}
    \caption{Classical complexity of P-ABD according to Nordh and Zanuttini (2008), illustrated on Post's lattice.}
    \label{fig:P-M-ABD}
\end{figure*}

\newpage

\bibliographystyle{plain}


\end{document}